\documentclass{article}

\PassOptionsToPackage{numbers, compress, sort}{natbib}

\usepackage{thm-restate}

\usepackage{chngcntr}
\usepackage{apptools}
\AtAppendix{\counterwithin{lemma}{section}}

\usepackage[utf8]{inputenc} 
\usepackage[T1]{fontenc}  
\usepackage{hyperref}    
\usepackage{url}      
\usepackage{amssymb}
\usepackage{booktabs}    
\usepackage{amsfonts}    
\usepackage{nicefrac}    
\usepackage{microtype}   
\usepackage[shortlabels]{enumitem}
\usepackage{amsthm}
\usepackage{color}
\usepackage[svgnames]{xcolor}
\usepackage{soul}
\usepackage{multirow}
\usepackage{fullpage}
\usepackage{layouts}
\usepackage{mathtools}

\usepackage{comment}
\usepackage{natbib}
\usepackage[ruled,vlined,noend]{algorithm2e}

\usepackage{array}
\usepackage{mathtools}

\usepackage[textwidth=20mm, textsize=footnotesize]{todonotes} 
\newcommand\bbE{\ensuremath{\mathbb{E}}} 
\newcommand\bbI{\ensuremath{\mathbb{I}}} 
\newcommand\bbP{\ensuremath{\mathbb{P}}} 
\newcommand\w{\ensuremath{\omega}} 

\theoremstyle{plain}

\newtheorem{lemma}{Lemma}

\theoremstyle{definition}
\newtheorem{definition}{Definition}

\hypersetup{
    colorlinks=true,
    linkcolor=blue,
    filecolor=blue,      
    urlcolor=blue,
    citecolor=blue
}

\newcommand{\Comments}{1}
\newcommand{\mynote}[2]{\ifnum\Comments=1\textcolor{#1}{#2}\fi}

\begin{document}

\title{Counterbalancing Learning and Strategic Incentives in Allocation Markets}

\date{\today}

\author{
	Itai Ashlagi\\
	Stanford University\\
	\texttt{iashlagi@stanford.edu} \\
	\and
	Jamie Kang\\
	Stanford University\\
\texttt{kangjh@stanford.edu}\\
\and
Moran Koren\\
Ben Gurion University of the Negev\\ \texttt{korenmor@bgu.ac.il}
	\and
	Faidra Monachou\\
	Yale University\\
	\texttt{faidra.monachou@yale.edu}}
\maketitle

\begin{abstract}

This paper considers the problem of offering  a scarce object with a common unobserved quality to strategic agents in a priority queue.  Each agent  has a  private signal over the quality  of the object and observes the decisions made by other agents.
We first show that, under the widely-used first-come-first-served sequential offering mechanism, herding behavior emerges: initial rejections create an information cascade resulting in  inefficient waste. To address this issue, we then introduce a class of batching mechanisms. Agents in each batch report whether they would be willing to accept or reject the object based on their private signals and prior information.  If the majority opts to accept, the object is randomly allocated within that batch. We prove that  suitable batching mechanisms are incentive-compatible and improve efficiency. 
A key property of the mechanism is the gradual  increase of the batch size after each failed allocation;  the size is chosen so that it elicits as much information as possible without distorting the incentives of agents to report truthfully. 
Additionally, from a healthcare policy perspective, our results can shed light on the large  wastage in organ allocation. In particular, wastage that arises due to herding may be reduced by applying adaptive simultaneous offering mechanisms. 
\end{abstract}

\section{Introduction}

Numerous scarce resources, such as cadaver organs, are allocated through a priority queue, where each object is sequentially offered to agents based on their priority. Such allocation mechanisms may suffer from misallocation and inefficient waste~\citep{mohan2018factors,agarwal2018dynamic}. In 2023, more than 25\% of kidneys recovered from deceased donors were discarded~\citep{OPTN_2023}.

 One contributing factor to the high discard rate is herding  \citep{Zhang2010,de2020herding}. Herding occurs when  refusals by preceding patients in the queue trigger a self-reinforcing chain of subsequent declines.  Under the ubiquitous first-come-first-served sequential offering mechanism,
initial rejections lead to an information cascade, resulting in inefficient waste. This paper offers a theoretical treatment by exploring when herding can arise and how we can alleviate it to improve efficiency. Specifically, we study the allocation problem when information about the object's quality is dispersed among rational agents who learn from the decisions of others.

To study this problem we consider a stylized model where the object to be offered has a common, ex ante unknown value, and can be either good or bad.  Each agent receives a private signal about the object's quality. When an agent is offered an object, she decides whether to accept or decline the offer based on the public signal (i.e., the information known to all agents before the game takes place)\footnote{In the context of deceased donor kidney allocation, the public signal of our model mimics the Kidney Donor Risk Index (KDRI) used in the US deceased donor kidney allocation process. The KDRI, calculated based on donor and organ characteristics, is made available to decision-makers at the start of the offering process.}, 
her private signal, and the  decisions made by other agents. Thus an offering mechanism induces a social learning  process among agents.  We study offering mechanisms that account for agents' incentives and seek to maximize efficiency measured by the  correctness of the allocation.

Our  model captures the social learning process among rational agents. 
First, we find that sequentially offering  the object to agents  can  result in poor correctness; 
independently of the true object quality, rejections  by few  agents at the head of queue  generates a cascade of rejections by all the subsequent agents.

To  address this issue, we introduce a class of batching mechanisms.  A \textit{batching mechanism} dynamically partitions the list into batches of varying sizes that may depend on the history of offers and the public signal. The mechanism   elicits responses from agents in each batch. All agents in the batch decide whether to provisionally opt in or opt out; 
if the majority of agents in the batch decides to opt in,  the object is  randomly allocated to one of the agents in the batch who decided to opt in.  If the list is exhausted without successful allocation, the object is discarded.

.

We show that, if a private signal is sufficiently informative relatively to the common prior, there always exists an incentive-compatible batching mechanism that strictly improves correctness. Furthermore,  we suggest a simple, greedy algorithm to implement such a mechanism. 
Otherwise, no  batching mechanism is incentive-compatible, and therefore no  batching mechanism can achieve higher correctness than the sequential offering mechanism.

Incentive-compatible batching mechanisms have several interesting properties. On the one hand, the batch size increases after each unsuccessful attempt to allocate the object. This is intuitive. Failing to allocate the object in a given batch results in a more pessimistic  belief about the object quality. 
Therefore the mechanism increases the batch size to ensure that the object is allocated only if sufficiently many agents receive positive signals.   
On the other hand, while more private signals can improve accuracy, the  batch size must be upper-bounded 
to maintain incentive-compatibility. This cap ensures that  agents do not disregard their private information and rely solely on the ``wisdom of the crowd."

We find that  continuing to offer the object to an additional  batch after a failed attempt strictly improves correctness. 
As our simulations illustrate,  this improvement can be significant and generally becomes larger as the prior and  signal precision grow. Nevertheless,  a batching mechanism with just two batches significantly outperforms sequential offering. Additionally, we find that the prior and signal precision have opposite effects on correctness: the correctness tends to decrease as the prior increases while  the marginal effect of signal precision is positive. 
Intuitively, as the prior increases, the optimal batch size in an incentive-compatible mechanism decreases.   This in turn implies that the  probability of incorrectly allocating a bad object increases. 
On the other hand, as the signal precision improves, the magnitude of the error decreases since  agents' signals become more reliable.

Our results highlight the tension between agents' strategic incentives and the planner's learning goal.
From the planner's perspective, agents' private signals contain valuable information about the unknown quality of the object. By increasing the number of collected data points, the planner can improve her confidence about the true quality of the object. 
Agents'  incentives, however,  impose a constraint on the maximum size of the batch that allows the truthful elicitation of agents' signals: 
as the batch size increases, the chance of an agent's reported signal to be pivotal decreases, thereby reducing the incentive for an agent with a negative signal to report it truthfully.
Our proposed batching mechanisms enable the planner to crowdsource agents' private information in a simple, truthful, and effective manner.

\smallskip \noindent\textbf{Related literature.}
Our model builds upon the classic social learning problem introduced by \citet{banerjee1992simple} and \citet{bikhchandani1992theory}. These seminal papers demonstrate that when agents receive binary signals about an object's quality and observe the actions of previous agents, information cascades—a phenomenon where agents disregard their private information and follow the actions of others—inevitably occur. Subsequent research has extended these findings to general signal distributions \citep{smith2000pathological} and more complex game structures, such as networks \citep{Acemoglu2010, Mossel2015}, multidimensional learning \citep{Arieli2019}, monopoly pricing \citep{crapis2017monopoly, papanastasiou2017dynamic}, and online reviews \citep{besbes2018information, acemoglu2017fast}.

However, these works do not consider the presence of a central planner who can control the information flow and optimize a global objective, such as the correctness of decisions \citep{arieli2018one}. The study most closely related to ours is \citet{deboModelsHerdingBehavior2009}, which examines a queueing system where common value objects are offered to agents. While \citet{deboModelsHerdingBehavior2009} focuses on queue length and its effect on agent decisions, showing that longer queues or stock-outs can increase demand, our work emphasizes allocation efficiency and the aggregation of private information.

In addition to the theoretical literature, several empirical studies investigate the presence of social learning in transplantation markets \citep{Zhang2010, de2020herding}. Using data from organ transplantation in the United Kingdom, \citet{de2020herding} develops empirical tests to detect herding behavior and quantify its welfare consequences. Similarly, \citet{Zhang2010} provides evidence of herding behavior in the United States deceased donor waitlist. Furthermore, a growing body of clinical literature \citep{butler2020behavioral, cooper2019report} examines the role of behavioral factors in organ allocation.

We propose a batching mechanism to enhance the allocation system's performance by eliciting agents' private information.  In general, as shown in the Condorcet Jury Theorem \citep{Condorcet1785}, batching procedures can reveal the ground truth by aggregating dispersed information. However, as \citet{Austen-Smith96} demonstrate, Condorcet's result does not hold when agents are strategic. Moreover, it does not apply to large collective purchases \citep{arieli2023information} and crowdfunding \citep{arieli2018one}, where a pre-registration phase allows consumers to decide whether to commit to purchasing, with production occurring only if enough registrations are made. \citet{arieli2023information, arieli2018one} show that in large consumer groups, a consumer receiving a discouraging signal may still commit to purchase, relying on the ``wisdom of the crowd," reducing the likelihood of achieving the correct collective decision. Similar mechanisms lead to the bounded batch size result in our model, which extends these results to an allocation problem with a single, indivisible object and examines the results in a sequential setting.

	\section{Model}
	\label{section: model}

	We consider a model in which a social planner seeks to allocate a single indivisible object of unknown quality to privately informed agents in a queue.  
For simplicity, we make two limiting assumptions: (1)  agent utilities are symmetric\footnote{Introducing heterogeneous agents adds little insight to our qualitative conclusions; 
therefore, for the sake of brevity, we considered homogeneous agents.};
(2) agents differ only in the realization of their private signals.\footnote{ We believe that our results are qualitatively robust and can be extended to a model with variations in agent types and asymmetric utilities.  Increasing the relative magnitude of accepting a ``bad'' object will make agents more cautious, thus leading to larger batches.} 
While our model is a simplified representation and does not incorporate all the complexities of organ allocation, it nonetheless offers useful insights. In Section \ref{sec:dis}, we explore the implications of our findings for the real-world context of deceased donor kidney allocation and discuss how our results can inform policy and decision-making in this critical area.

	\textbf{Object and agents.}
	The object is characterized by a fixed \textit{quality} $\w \in \{G, B\}$, where $G$ and $B$ denote a \textit{good} and \textit{bad} quality, respectively. 
	The true quality $\w$ of the object is ex-ante unknown to both the planner and agents; however, both share a common prior belief ${\mu=\bbP(\w = G) \in (0,1).}$

	Agents are waiting in a queue; we denote by $i$ the agent in position $i.$ Each 
	 agent  knows his own position.  For simplicity, we assume that the queue consists of an arbitrarily large number $I$ of agents.
	Each agent $i$ has a private binary \emph{signal} $s_{i}\in \mathcal{S} \triangleq \{g,b\}$ that is informative about the true quality  $\w$ of the object. Each $s_i$ is identically and independently distributed  conditional on $\omega.$ Furthermore, the signal value $s_i$ is aligned with the true object quality $\w$ with probability
	${q = \bbP(s_{i}= g \mid \w = G) = \bbP(s_{i}= b \mid \w = B),}$
	where  the signal precision $q\in(0.5,1)$ is  commonly known.\footnote{The condition $q \in (0.5, 1)$ is without loss of generality. The common prior belief $\mu$ and signal precision $q$ are  standard  in the social learning literature~\citep{banerjee1992simple, bikhchandani1992theory,  smith2000pathological}.}
	
	Agents are risk-neutral. The utility of the agent who receives the object is $1$ if  $\omega=G$ and $-1$ if $\w = B$. (This symmetry simplifies the exposition and analysis but does not qualitatively change the results.)  
	Any agent who does not receive the object, 
	because he either declines or is never offered,
	receives a utility of 0. As a tie-breaking rule, we assume that indifferent agents always decline.

\textbf{Batching mechanisms.} 
The planner designs a mechanism which potentially asks (a batch of) agents to report their private signals, and based on their report, decides whether and how to allocate the object. 
We consider a class of  batching mechanisms denoted by $\mathcal{V}.$ A batching mechanism $V \in \mathcal{V}$ offers the object sequentially to odd-numbered batches of agents, where each batch size may be set dynamically based on the information from  previous batches. In any given batch, if the majority of agents vote to opt in, the object is allocated uniformly at random to one of the current agents who opted in. 

Formally, a {\it batching mechanism} $V_{\{\pi_j\}_{j=1}^\infty}$ is defined by a  sequence of functions $\pi_1,\pi_2,...$, such that for each batch $j,$ the corresponding function  $\pi_j: [0,1] \rightarrow \mathbb{N}$  maps a current  belief $\mu_{j-1}$  about the object quality $\omega$ to a batch size $K_j$ (we let $\mu_0=\mu$).  The mechanism begins by offering the object to the first batch of agents, who are in positions $1,\ldots,\pi_1(\mu)$. If the object is not allocated to batch  $j-1,$ it is subsequently offered to the  agents in the next batch $j$, i.e., agents in positions $\pi_{j-1}(\mu_{j-2}) +1, \ldots, \pi_j(\mu_{j-1})$.  

When agents in batch $j$ are offered the object, each agent $i$ in batch $j$ chooses an action (a vote) $\alpha_i \in \mathcal{A} \triangleq \{y, n\}$, where $y$ and $n$ correspond to \textit{opting in} (``yes'') and \textit{opting out} (``no''), respectively. If $Y_j,$ the number of opt-in votes from batch $j$, constitutes the \textit{majority} of batch $j$, i.e., $Y_j \ge \lceil \frac{\pi_j(\mu_{j-1})}{{2}} \rceil$, 
then the object is allocated uniformly at random among agents in batch $j$ who opt in. To simplify  exposition, we sometimes denote by $K_j$ the ex-post size of  batch $j$ (thus omitting the dependency on the current  belief $\mu_{j-1}$).
We assume that the agents in  batch $j$ as well as the planner observe all the votes of agents in the previous batches $j'<j$, and this is commonly known. After asking a batch $j$ of agents and given a prior $\mu_{j-1}$, the  belief $\mu_j$ shared by the planner and agents is updated recursively in terms of the realized batch size $K_j$ and voting outcome $Y_j$ as
follows:
\begin{equation}
	\label{eq: prior_mu_j}
	\mu_j \triangleq \bbP\left(\omega \mid \mu_{j-1}, K_j, Y_j \right)
	= \frac{\mu_{j-1}q^{Y_j}(1-q)^{K_j-Y_j}}{\mu_{j-1}q^{Y_j}(1-q)^{K_j-Y_j} + (1-\mu_{j-1})q^{K_j-Y_j}(1-q)^{Y_j}}.
\end{equation}

A common batching mechanism in practice  is the {\it sequential offering} mechanism, which offers the object to one agent at a time. Another  type of batching mechanism 
is {\it single-batch} mechanisms, which offer the object only once; thus, the object is immediately discarded if the majority in this single batch choose to opt out.

Let $u_i (\alpha_i; s_i)$ denote the expected utility of an agent who receives signal $s_i \in \mathcal{S}$ and takes action $\alpha_i \in \mathcal{A}$.
A batching mechanism  $V \in \mathcal{V}$ is \textit{incentive-compatible} (IC) if
$u_i (y; g) \ge u_i (n; g)$ 
and $u_i (n; b) \geq u_i (y; b).$\footnote{Note that we assume that indifferent agents always decline.} 
Thus, agent $i$ gets higher utility from opting in than opting out when he has received signal $s_i=g$, but prefers to opt out when $s_i=b$.

\textbf{Correctness.}
The planner's goal is to maximize the probability of a correct allocation outcome, that we will refer to as \textit{correctness}.  The allocation outcome is correct if it allocates a good object or does not allocate a bad object. 
Formally, under a batching mechanism $V \in \mathcal{V}$, the planner's decision whether to allocate the object or not  is denoted by the random variable $Z=\{0,1\}.$
	We say that a  mechanism $V \in \mathcal{V}$ achieves \textit{correctness} $c(V)$ if
	\begin{align*}
		\label{eq: correctness_definition}
		c(V) \triangleq 
		\bbE_V \left[ \bbI(\w=G \, \cap 
		\, Z=1) + \bbI(\w=B \, \cap 
		\, Z=0) \right].
	\end{align*}

	Note that, in our model where a single object is allocated and all agents are identical, a cascade can only occur with the `opt-out' action. Therefore, maximizing correctness is equivalent to maximizing social welfare.

	\textbf{Upper bound on correctness.} 
	A natural benchmark on correctness would be the optimal solution in a setting where agents are not strategic and thus are willing to truthfully reveal their private signal to the planner.
	In the absence of strategic incentives, the planner's optimal solution would be to (i) ask \textit{all} agents in the queue to reveal their private signals, (ii) compute her posterior belief based on the gathered information, and (iii) allocate the object to a random agent if and only if her posterior exceeds 0.5. Equivalently, as we show in Lemma~\ref{lemma: optimal_threshold_no_incentives} in Appendix~\ref{appendix: no-incentives}, the planner allocates the object if and only if the number of positive signals meets a specific threshold $\underbar{y}$ defined as:
		$${\underbar{y} \triangleq 0.5 \, \log \left(\frac{1-\mu}{\mu}\right) \left (\log \left(\frac{q}{1-q}\right)\right)^{-1} +0.5 \, I},$$
	where $I$ is the number of agents in the queue. 
		In that case, we can show that correctness equals
	\begin{equation}
		\label{eq: correctness_upper_bound}
		\bbP \left( X_I \ge \underbar{y}  \right),
	\end{equation}
	where $X_I$ is a Binomial$(I, q)$ random variable (see also Lemma~\ref{lemma: existence-and-correctness}). As $I$ grows to infinity, correctness approaches 1.
	Importantly, in the presence of strategic incentives, this upper bound is not achievable by any incentive-compatible mechanism. 

\textbf{Discussion of the assumptions.}
We assume that the outcomes after each batch are fully revealed. This assumption is done for simplicity. Dropping this knowledge would complicate the analysis substantially (without affecting our the major qualitative insights): one should compute the posterior belief of each agent  by taking expectations of all the possible outcomes that led to them being offered an object.

Our model focuses on the allocation of a single object rather than a setting in which multiple objects arrive and are allocated over time. In particular,  we intentionally abstract  away from  addressing dynamic incentives and restrict  attention to the learning  aspect of allocation. 

	\section{Herding in the Sequential Offering Mechanism}
	\label{section: seq_benchmark}

	We begin by analyzing the commonly applied sequential offering mechanism, which will serve as a  benchmark for our results.
	Recall that the \emph{sequential offering} mechanism, namely $V_{\textsc{seq}}$, offers the object to each agent one-by-one; the first agent that is offered the object and opts in, receives it.

	As we  establish below, $V_{\textsc{seq}}$ has two important  drawbacks: first, it is  not incentive-compatible; second, it achieves poor correctness.
	This is due to a 
	cascade of actions (i.e., herding) that takes place after a few decisions. Specifically, a cascade of opt-out actions begins after a couple of initial agents opt out, which ultimately leads to the discard of the object. The number of initial opt-out actions needed to instigate this cascade-led discard depends on the prior  $\mu.$ The following result formalizes this. 
	(A detailed analysis of the benchmark mechanism can be found in Appendix~\ref{appendix: agents-sequential-offering} together with the omitted proofs of the section.)
	
	\begin{restatable}{lemma}{lemmasequentialallocation}
		Under 
		$V_{\textsc{seq}}$, the object is allocated if and only if: (i) $\mu > q$, (ii) $\mu \in [1/2,q]$ and either $s_1=g$ or $s_2=g$, (iii) $\mu \in [1-q,1/2]$ and $s_1=g.$ Otherwise, the object is discarded.
		\label{lemma: sequential-allocation}
	\end{restatable}

	\begin{figure}[h!]
		\centering
		\includegraphics[width=0.9\textwidth]{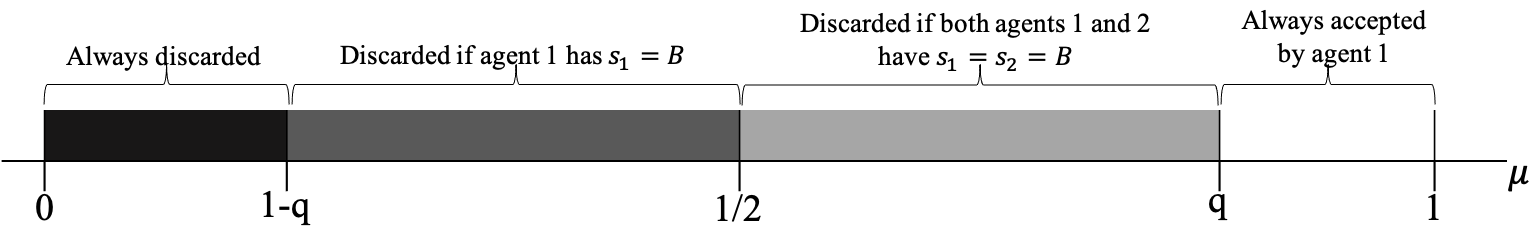}
		\caption{Allocation outcome of the sequential offer mechanism  ($V_\textsc{seq}$) based on the value of prior $\mu \in (0,1)$ with respect to  signal precision $q.$ For low $\mu,$ the object is rejected by all agents, leading to discard. For intermediate $\mu,$ the allocation outcome depends on the initial two agents' private signals.
		} 
	\label{figure: sequential-offering}
\end{figure}
A direct implication of this result is that the  outcome of the mechanism is determined in the first two offers. Regardless of the values of $\mu$, $q$ and the object's true quality, the object is always discarded if it is rejected by the first two agents. Thus the result below follows.

\begin{restatable}{proposition}{propositionseqcorrect}\label{prop: sequential-two-agents}
	Under the sequential offering mechanism, $V_{\textsc{seq}}$:
	\begin{itemize}
		\item[(i)] At most the two initial agents in the queue determine the allocation outcome: if the first two agents decline, then the object will be always discarded. 
		\item[(ii)] The correctness equals
		 \begin{align*} 
			c(V_{\textsc{seq}})=
			\begin{cases} 
				\mu &\text{ for } \mu > q\\
				2\mu q(1-q)+q^2 &\text{ for } \mu \in [1/2,q]\\
				q &\text{ for } \mu \in [1-q,1/2)\\
				1-\mu &\text{ for } \mu < 1-q.
			\end{cases}
	 \end{align*}
	\end{itemize}	
\end{restatable}
 
 The result qualitatively suggests that the first few\footnote{We have assumed that agents share a constant prior $\mu.$ In practice, each agent's prior might be drawn i.i.d. from some common distribution with support $(0,1)$. Our qualitative insights still extend to this case: a larger but finite number of initial agents determines the allocation outcome and herding still occurs.
} agents in the queue  have the power to decide whether the object will be allocated. This implies that these few agents 
can inadvertently undermine the welfare of the rest of the agents. Most importantly, the planner can never learn from the private information of the remaining $I-2$ agents to make better 
allocation
decisions. In the context of organ transplants, this herding behavior prevents the planner from allocating good-quality organs, leading to a high discard rate. This may in turn lead to patients' longer waiting times and health deterioration. 

 Motivated by the drawbacks of $V_{\textsc{seq}}$, in the next section, we will show how batching mechanisms can improve correctness.

\section{Batching Mechanisms}
\label{section: voting-mechanisms}

Recall from Equation~\eqref{eq: correctness_upper_bound} that in the absence of strategic incentives, the planner asks as many agents as possible and makes her allocation decision based on these (truthful) private signals. On the other hand, 
when agents are strategic, the presence of their incentives introduces a constraint for the planner, which in turn sets an upper bound on the number of agents the planner can truthfully learn from. 
Under the benchmark mechanism $V_\textsc{seq},$ for example, the planner can elicit at most two truthful signals, which may prevent her from taking a correct allocation decision. 

Motivated by these shortcomings, 
we propose a new class of batching mechanisms and examine the incentive-compatible mechanisms within this class.
We show that there always exists an incentive-compatible batching mechanism that improves correctness as long as private signals are more informative 
than the common prior belief; otherwise, none of the batching mechanisms are incentive-compatible and 
all achieve correctness equal to that of $V_\textsc{seq}$. 
Theorem \ref{theorem: main-result} is our main result. 

\begin{restatable}{theorem}{thmmainresult}
	For any $\mu < q,$ there exists a batching mechanism $V \in \mathcal{V}$ that is incentive-compatible and improves correctness  compared to the sequential offering mechanism $V_{\textsc{seq}}$.
	For $\mu \geq q,$ there is no incentive-compatible  batching mechanism and any $V \in \mathcal{V}$ achieves the same correctness as  
	$V_\textsc{seq}$.
	\label{theorem: main-result}
\end{restatable}

We prove Theorem \ref{theorem: main-result} in several steps. In Section~\ref{section: single-batch}, we begin by defining the simplest batching mechanisms, the {\it single-batch mechanisms,} where the mechanism offers the object to one batch only and terminates afterwards. Specifically, we characterize the optimal mechanisms among all incentive-compatible single-batch batching mechanisms. Then in Section~\ref{section: multi-batch}, we utilize these results to develop a simple greedy algorithm to construct an incentive-compatible, correctness-improving batching mechanism with potentially \textit{multiple} batches.

\subsection{Warm up: Optimal Single-batch Batching Mechanisms}
\label{section: single-batch}

As a warm up, we begin with the study of one of the simplest batching mechanisms: the single-batch  mechanism. With a slight abuse of notation, we use the following definition.

\begin{definition}
	Let $V_{\{K\}} \in \mathcal{V}$ be the {\it single-batch  mechanism} where the object is offered to only one batch of the initial $K$ agents in the queue.
\end{definition}

Unlike $V_\textsc{seq}$, mechanism $V_{\{K\}}$ (as well as any other batching mechanism $V \in \mathcal{V}$) ensures that the object is allocated if and only if the majority of agents (in the batch) opt in. As our results in this section illustrate, the extent to which this policy incentivizes agents to vote truthfully depends on the values of the prior $\mu$ and batch size $K$.

In Lemma~\ref{lemma: ic-combined}, we identify the necessary and sufficient conditions for $V_{\{K\}}$ to be incentive-compatible. In Lemma~\ref{lemma: existence-and-correctness}, we show its existence conditional on the informativeness of signals and characterize the correctness it achieves. In Proposition~\ref{prop: single-optimal}, we establish that among all incentive-compatible, single-batch mechanisms $V_{\{K\}},$ correctness is maximized when the incentive compatibility constraint is binding at $K=\overline{K}(\mu).$ Finally, in Lemma~\ref{lemma: comparative}, we provide simple comparative statics of this optimal batch size. In the next section, we will extend these results to batching mechanisms with multiple batches.

First, for any batch size $K,$ let 
\begin{equation*}
	{\mathcal{I}_K \triangleq 
		\left( 
		1 - \bbP\left(X_{K} \ge \frac{K+1}{2}\right), \frac{q^2 \left(1 - \bbP\left(X_{K} \ge \frac{K+1}{2}\right)\right)}{q^2 \left(1 - \bbP\left(X_{K} \ge \frac{K+1}{2}\right)\right) + (1-q)^2 \bbP\left(X_{K} \ge \frac{K+1}{2}\right)}    \right)}
\end{equation*}
be some interval of the prior where ${X_K\sim\text{Binomial}(K,q)}.$ Moreover, let $\overline{K}(\mu)$  (respectively, $\underline{K}(\mu)$) be the maximum (respectively, minimum) batch size $K$ such that $\mu \in \mathcal{I}_K.$ Then the incentive compatibility of $V_{\{K\}}$ can be characterized in terms of the interval $\mathcal{I}_K,$ or equivalently, the batch size bounds $\overline{K}(\mu)$ and $\underline{K}(\mu)$ as in Lemma~\ref{lemma: ic-combined} below.

\begin{restatable}{lemma}{lemmaiccombined}
	$V_{\{K\}}$ is incentive-compatible if and only if
	(i) $\mu \in \mathcal{I}_K$, or equivalently (ii) $\mu < q$ and $\underline{K}(\mu) \le K  \le \overline{K}(\mu).$ 
	\label{lemma: ic-combined}
\end{restatable}

The proof of the lemma can be found in
Appendix~\ref{appendix: single_batch_proofs} together with the rest of the proofs in this section. 
A sketch of the proof is as follows. To show condition~(i), for a fixed $K$,   let $\mathcal{G}_K$ (respectively, $\mathcal{B}_K$) be the probability that the object gets allocated to some agent in the batch conditional on the true quality being good (respectively, bad). We compute that
\begin{align*}
	{\mathcal{G}_K = \sum_{y=\frac{K+1}{2}}^K \frac{1}{y} {K-1 \choose y-1} q^{y-1}(1-q)^{K-y}}, &\qquad 
	{\mathcal{B}_K = \sum_{y=\frac{K+1}{2}}^K \frac{1}{y} {K-1 \choose y-1} q^{K-y}(1-q)^{y-1}.}
\end{align*}
Next, we use $\mathcal{G}_K$ and $\mathcal{B}_K$ to rewrite the incentive compatibility constraints as 
 \begin{equation*}
	 \label{eq: IC_function_mu_BK_GK}
	 \begin{split}
		{\mu q \mathcal{G}_K - (1-\mu)(1-q)\mathcal{B}_K \ge 0} \textrm{~~and~~} {\mu (1-q) \mathcal{G}_K - (1-\mu)q \mathcal{B}_K \le 0.}
		 \end{split}
	 \end{equation*}
Given that the left-hand sides of both inequalities monotonically increase with $\mu,$ there exist some thresholds of prior, namely $\underline{\mu}_K \in (0,1)$ and $\overline{\mu}_K \in (0,1)$, that solve the indifference conditions:
 \begin{align*}
	\frac{\underline{\mu}_K}{1-\underline{\mu}_K} = \frac{1-q}{q} \frac{\mathcal{B}_K}{\mathcal{G}_K} \textrm{~and~} \frac{\overline{\mu}_K}{1-\overline{\mu}_K} = \frac{q}{1-q} \frac{\mathcal{B}_K}{\mathcal{G}_K}.
	 \label{eqn: upper-and-lower-mu-and-B/G}
	 \end{align*}
Based on this transformed system of equations, we are able to characterize the feasible region $(\underline{\mu}_K, \overline{\mu}_K)$ for the prior $\mu$ which turns out to be precisely $\mathcal{I}_K$. To show the equivalent condition~(ii), we utilize some useful properties of the interval $\mathcal{I}_K$: first, 
its endpoints are weakly decreasing in $K$ (Lemma~\ref{lemma: ic-interval-decrease}), and second, it is overlapping (i.e., $\mathcal{I}_K \cap \mathcal{I}_{K+2} \neq \phi$) (Lemma~\ref{lemma: ic-intervals-overlap}). 

\begin{figure}[]
	\centering
	\includegraphics[width=0.85\textwidth]{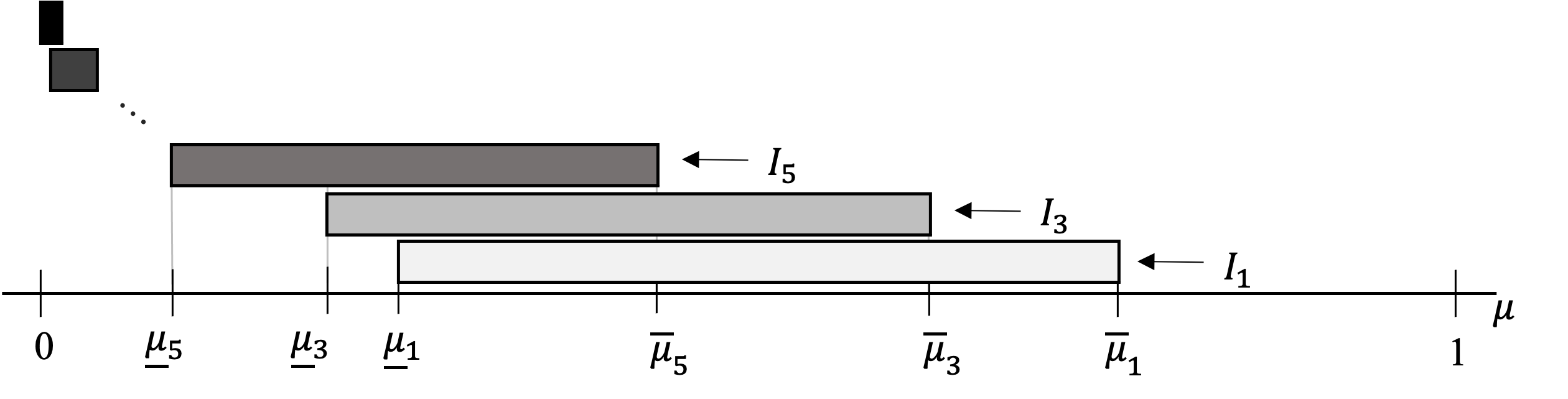}
	\caption{Interval of incentive-compatible prior $\mathcal{I}_K$ for different batch sizes $K \in \{1, 3, \ldots\}$ as described in Lemma~\ref{lemma: ic-combined}. The intervals $\mathcal{I}_K$ are decreasing (Lemma~\ref{lemma: ic-interval-decrease}) with $K$ and the consecutive intervals $\mathcal{I}_K$ and $\mathcal{I}_{K+2}$ overlap (Lemma~\ref{lemma: ic-intervals-overlap}) with each other.}
	\label{fig: feasible-priors}
\end{figure}

Condition~(ii) in terms of batch size $K$ provides a straightforward intuition. Larger batch sizes make agents more confident that, if the object is allocated, it is likely that the quality is good. For the same reason, however, a batch size that is \emph{too large} incentivizes everyone to opt in, including those with negative signals. Borrowing the standard terminology from the voting literature, $\overline{K}(\mu)$ is designed to make each agent's decision \textit{pivotal}: $\overline{K}(\mu)$ is the maximum batch size that guarantees incentive compatibility. On the other hand, if the batch size is \textit{too small}, then the allocation decision
depends on learning from a sample size that is too small to offer reliable information, making agents with positive signals reluctant 
to opt in. Thus, $\underline{K}(\mu)$ is the minimum batch size that ensures incentive compatibility. 

Furthermore, Lemma~\ref{lemma: ic-combined} gives rise to the following observation:

\begin{restatable}{lemma}{lemmaexistencecorrectness}
	Suppose $\mu < q.$ Then, an incentive-compatible mechanism $V_{\{K\}}$ always exists and achieves correctness $c\left(V_{\{K\}}\right)=\bbP\left(X_K\ge\frac{K+1}{2}\right),$ where $X_K \sim \textrm{Binomial}(K,q).$
	\label{lemma: existence-and-correctness}
\end{restatable}

Lemma~\ref{lemma: existence-and-correctness} shows that if each signal is at least as informative as the prior $\mu$, there is a single-batch  mechanism that elicits information in a truthful manner. Another implication of the lemma is that the correctness of an incentive-compatible mechanism $V_{\{K\}}$ can be expressed as the tail distribution of a $\textrm{Binomial}(K,q)$ random variable. This tail distribution can be interpreted as the probability that the majority of votes are aligned with the true quality. Given the definition of correctness, this property is thus particularly intuitive.

An additional observation is that the tail distribution $\bbP\left(X_K \ge \frac{K+1}{2}\right)$ increases with $K$ (see Lemma~\ref{property: tail-increase} in Appendix~\ref{section: supp-proofs}). 
Consequently, given $\mu,$ the upper size limit $\overline{K}(\mu)$ maximizes correctness among all possible values of $K$ that conserve the incentive compatibility of $V_{\{K\}}.$ We formalize this property below. 

\begin{restatable}{proposition}{lemmaoptimalsinglebatch}
	Suppose $\mu < q.$ Then the batch size $K = \overline{K}(\mu)$ maximizes correctness among all incentive-compatible $V_{\{K\}}.$ 
		\label{prop: single-optimal}
\end{restatable}

Hence the planner will choose the largest batch size that the agents' incentives allow.
As we establish below, the optimal batch size behaves monotonically with respect to the prior $\mu.$ 

\begin{restatable}{lemma}{lemmacomparative}
	$\overline{K}(\mu)$ weakly decreases with $\mu.$ 
	\label{lemma: comparative}
\end{restatable}

The logic here is simple. If the sentiment around the object quality is optimistic (i.e., $\mu$ is high), then it is harder to keep the agents with negative signals truthful, requiring a tighter upper bound on the batch size (lower $\overline{K}(\mu)$). We discuss additional comparative statics through simulations in Section~\ref{sec.simulations}. 

	Finally, note that the results above focused on regimes where $\mu < q$. This a natural setting which guarantees that the signal is at least as informative as the prior belief. 
	For the complementary case where $\mu \geq q$, we show that there is no incentive-compatible batching mechanism (see Lemma~\ref{corollary: impossibility} in Appendix~\ref{appendix: single_batch_proofs}). 
	Intuitively, the lack of informativeness of the signal induces the agents to lose trust in their private signals, and as a result creates no incentives to behave truthfully. In that case, under any $V_{\{K\}},$ all agents will opt in and thus the object will always be allocated. Notice that this outcome is equivalent to the outcome under $V_\textsc{seq}$ in terms of correctness.\footnote{The difference is that under $V_{\{K\}},$ the object is always allocated to an agent in the batch uniformly at random, whereas under $V_\textsc{seq}$ it is always allocated to agent $1$.}

	\subsection{Improving Correctness via Greedy Batching Mechanisms}
	\label{section: multi-batch}

	Next we utilize the results of Section~\ref{section: single-batch} 
	to characterize how to dynamically choose batch sizes of a batching mechanism (potentially with multiple batches) to improve correctness, while balancing the agents' incentives.  
	Further, we implement this mechanism via a simple greedy algorithm.

	The first key technical observation is that, from the perspective of the agents in batch $j,$  the batching mechanism $V_{\{\pi_j\}_{j=1}^{\infty}}$ is equivalent to a single-batch  mechanism $V_{\{K_j \}}$ with realized common  belief $\mu_{j-1}$ and  batch size $K_j = \pi_j \left( \mu_{j-1} \right).$ Naturally, since both the planner and agents observe the voting results from previous batches, no informational asymmetry arises between them.
	
	\begin{restatable}{lemma}{lemmaicgreedy}
		\label{lemma: IC-greedy}
		A batching mechanism $V_{\{\pi_j\}_{j=1}^\infty}$ is incentive-compatible if and only if for any batch $j,$ current belief $\mu_{j-1}$ and batch size $K_j\triangleq\pi_j \left(\mu_{j-1}\right),$ 
		$V_{\{K_j \}}$ is incentive-compatible.
	\end{restatable} 
	Hence, to ensure the incentive compatibility of $V_{\{\pi_j\}_{j=1}^\infty},$ it is 
	sufficient to choose each batch size $K_j$ myopically by solving a single-batch  mechanism design problem (with updated belief $\mu_{j-1}$ instead of $\mu$). 
	This suggests the following greedy scheme.

	\begin{definition} 
		For any $J \in \mathbb{N},$ let $V_{\textsc{greedy}}^J \in \mathcal{V}$ be the batching mechanism that offers the object to $J$ batches unless either it is allocated or the list is exhausted. Formally,
		for each batch $j$ and  realized belief $\mu_{j-1},$ the batch sizes are chosen as
		 \begin{align*}
			K_j = \pi_{j}(\mu_{j-1}) = 
			\begin{cases}
				\overline{K}(\mu_{j-1}) &\text{ for } j \le J \\
				0 &\text{ for } j > J.
			\end{cases} 
			 \end{align*}
		\label{definition: greedy} 
			\end{definition}
	Building upon Lemma~\ref{lemma: IC-greedy} and  results from Section~\ref{section: single-batch}, we establish three key properties of $V_\textsc{greedy}^J$.  
	\begin{restatable}{proposition}{propjgreedy}
		For any $\mu <q$ and $J,$ $V_\textsc{greedy}^J$ has the following properties: 
		\begin{itemize}
			\item[(i)] $V_\textsc{greedy}^J$ is incentive-compatible;
			\item[(ii)] Ex-post batch sizes satisfy $K_{j'} \le K_{j}$ for any $j'<j \in [J]$;
			\item[(iii)]$c\left(V_\textsc{greedy}^J\right)$ strictly increases with $J$.
		\end{itemize}
		\label{prop: greedy-j-properties}
	\end{restatable}
	To understand part~(ii), 
	suppose that the object had been offered to batch $j$ but was not allocated. Then, the current belief naturally decreases, that is $\mu_j < \mu_{j-1}$, as the majority of batch $j$ has voted to opt out. 
	By Lemma~\ref{lemma: comparative}, a 
	more pessimistic belief about the object quality results in a larger optimal batch size $\overline{K}(\cdot).$ Hence, it follows that $\overline{K}(\mu_{j-1}) \le \overline{K}(\mu_{j}).$ Part (iii) of the proposition suggests
	that adding an additional batch improves correctness. The intuition is that an additional batch allows the planner to collect and learn from more information about the object. 
	Consequently,
	the planner would like to keep on offering until either the object is allocated or there are not enough agents left in the queue. We formally define this mechanism below and describe a simple algorithm to implement it. 
	
	\begin{definition}
		Let $V_{\textsc{greedy}} \in \mathcal{V}$ be the \emph{greedy batching mechanism} such that for each batch $j$ and the realized belief $\mu_{j-1},$ the batch sizes are chosen as $K_j = \pi_j(\mu_{j-1}) = \overline{K}(\mu_{j-1}).$
	\end{definition}
	\footnotesize
	\begin{algorithm}[H]
		\SetAlgoLined
		Initialize batch $j=1,$ belief $\mu_0 = \mu,$ batch size $K_1 = \overline{K}(\mu_0)$\;
		\While{$I \ge K_1 + \ldots + K_j$}{
			Collect votes from the top $K_j$ remaining agents\;
			\eIf{$Y_j \ge \lceil \frac{K_j}{2} \rceil$}{
				Allocate the object uniformly at random among agents in batch $j$ who opted in\;
			}{
				Update belief $\mu_j$ using Equation~\eqref{eq: prior_mu_j} and next batch size as $K_{j+1} = \overline{K}(\mu_{j})$\;
				Set $j \leftarrow j+1$\;
			}
		}
		\caption{Implementation of $V_\textsc{greedy}$}
		\label{alg: greedy}
	\end{algorithm}
	\normalsize

 Note that a major advantage of $V_\textsc{greedy}$ is its design simplicity: the maximum number of batches, $J,$ does not have to be predefined.
	
	Finally, we are ready to prove Theorem~\ref{theorem: main-result}. The main idea is that 
	using the correctness $c(V_\textsc{seq})$ of the sequential offering mechanism from Proposition~\ref{prop: sequential-two-agents},
	we can show that for every $\mu < q,$  the planner can achieve $c\left(V_\textsc{greedy}^J\right)$ that exceeds $c(V_\textsc{seq}),$ by appropriately setting $J$ (see the full proof in Appendix \ref{appendix: main-theorem-proof}). Furthermore, using Proposition~\ref{prop: greedy-j-properties} (part (iii)), we establish that the improvement described in Theorem~\ref{theorem: main-result} can be also achieved by $V_\textsc{greedy}.$ 
	
	\begin{restatable}{corollary}{corollarygreedy}
		For any $\mu < q,$ $V_\textsc{greedy}$ is an incentive-compatible (multi-batch)  mechanism that improves correctness in comparison to the sequential offering mechanism $V_{\textsc{seq}}$.
	\end{restatable}

\section{Simulations}

\label{sec.simulations}
In this section we use simulations to complement our theoretical results.

\textbf{Optimal batch size.} 
In Figure~\ref{fig:optimal_batch}, we study the optimal batch size $\overline{K}(\mu)$ (Proposition~\ref{prop: single-optimal}) of the single-batch  mechanism $V_{\{K\}}$; 
by construction,
the optimal  $V_{\{K\}}$ is equivalent to $V_\textsc{greedy}^1.$ In our numerical analysis, we examine $\overline{K}(\mu)$ for all priors $\mu \in (0,1)$ and signal precision ${q \in \{0.6, 0.7, 0.8\}.}$

We make several observations. First, as guaranteed by Lemma~\ref{lemma: comparative}, the optimal batch size $\overline{K}(\mu)$  decreases as the prior $\mu$ increases. As we discussed in greater detail in Section~\ref{section: single-batch}, 
the incentive compatibility constraint binds at lower values of $\overline{K}(\mu).$ For very high values of $\mu$, that is $\mu \geq q$, Theorem~\ref{theorem: main-result} establishes that it is impossible to achieve incentive compatibility for any batch size, which explains the discontinuity at $\mu=q$ in each curve.

Second, for lower values of $\mu$ (approximately for values $\mu < 0.5$), the marginal effect of $q$ on $\overline{K}(\mu)$ is negative but diminishes as $\mu$ grows. 
Indeed, as Figure~\ref{fig:optimal_batch} illustrates,  the signal precision $q$ plays an important role for lower values of $\mu$. In particular, if $\mu$ is low and $q$ is not significantly informative (e.g., $q=0.6$), the planner has a priori low confidence in the object's quality. Similarly, agents have strong incentives to opt out regardless of their  signal.
Thus, the planner needs a larger batch size.

\begin{figure}
	\begin{minipage}{0.45\textwidth}
		\includegraphics[width=0.9\textwidth]{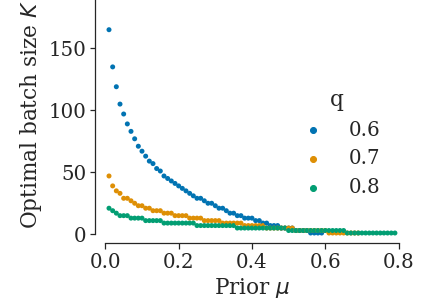}  
	\end{minipage}
	\begin{minipage}{0.54\textwidth}
		\normalsize \caption{Via simulations, we compute the optimal incentive-compatible batch size $\overline{K}(\mu)$ for all possible priors $\mu \in (0,1)$ for three regimes: $q \in \{ 0.6, 0.7, 0.8\}$. In all regimes, $\overline{K}(\mu)$ decreases as $\mu$ increases. For low values of $\mu$, higher $q$  implies lower $\overline{K}(\mu)$.}
		\label{fig:optimal_batch}
	\end{minipage}
\end{figure}
\begin{figure}
	\centering
	\includegraphics[width=\textwidth]{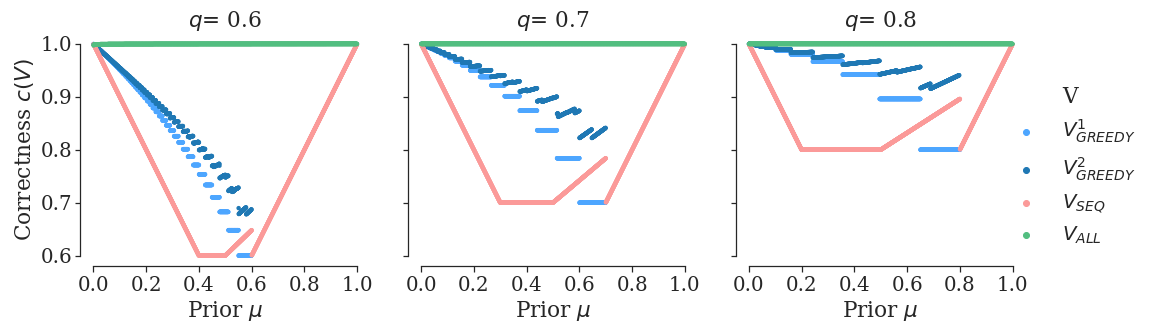}
	\caption{We compare the correctness $c(V)$ of  mechanisms ${V \in \{V_\textsc{greedy}^1, V_\textsc{greedy}^2, V_\textsc{seq}, V_\textsc{all}\}}$ in a setting with population size $I=345.$ $V_\textsc{all}$ represents the optimal mechanism  in the absence of strategic incentives and achieves near-optimal correctness. 
 The remaining $V_\textsc{greedy}^1$, $V_\textsc{greedy}^2$, and $V_\textsc{seq}$  assume strategic agents. We compute $c(V)$ for all  $\mu \in (0,1)$ and three different regimes $q \in \{0.6,0.7,0.8\}.$ In all regimes, $V_\textsc{greedy}^1$ achieves higher correctness than $V_{\textsc{seq}}$ for all $\mu < {q}/{2} + {1}/{4}$.
		With one additional batch, $V_\textsc{greedy}^2$ further improves the correctness of $V_\textsc{greedy}^1$ and outperforms $V_{\textsc{seq}}$ for all $\mu < q$.
	} 
	\vspace*{-.2cm}
	\label{fig: comparison}
\end{figure}

\textbf{Correctness evaluation and comparison.} In Figure~\ref{fig: comparison}, we evaluate the correctness $c(V)$ of four different mechanisms $V \in \{V_\textsc{greedy}^1, V_\textsc{greedy}^2, V_\textsc{seq}, V_\textsc{all}\}$ for $q\in \{0.6, 0.7, 0.8\}$ and  $\mu \in (0,1)$. In particular, $V_\textsc{all}$ represents the optimal mechanism  in the absence of strategic incentives (see Equation~\eqref{eq: correctness_upper_bound}).

For $q\in\{0.6, 0.7, 0.8\}$, we  study the behavior of $c(V^J_{\textsc{greedy}})$, $J \in \{1,2\}$, as a function of the prior $\mu$. We observe several interesting properties. First,  $\mu$ and $q$ apply two opposite forces on correctness: for both $J \in \{1,2\}$, $c(V^J_{\textsc{opt}})$ tends to decrease as $\mu$ increases but for higher $q$'s this effect is smaller.  
Intuitively, 
as the prior $\mu$ grows, the batch size becomes smaller (see Figure~\ref{fig:optimal_batch}). As the planner learns from fewer signals,
the probability to misallocate the object increases. 
However, as the signal precision $q$ improves,  the likelihood of misallocating the object decreases because the signals become more reliable.

Second, the comparison between $V_\textsc{greedy}^1$ and $V_\textsc{greedy}^2$ confirms that an additional
batch has an important positive effect on correctness. The two-batch mechanism $V_\textsc{greedy}^2$ outperforms the single-batch mechanism $V_\textsc{greedy}^1$. In fact, the gap between their achieved correctness grows as $\mu$ increases. Intuitively,  if the planner fails to allocate a good object in the first batch, she has 
another chance to allocate it in the second batch.
Since a higher $\mu$ translates to a higher probability that the object is of good quality, having
two opportunities to allocate the object has a positive impact on correctness. This also explains the non-monotonic---though generally decreasing---trend of $c(V_\textsc{greedy}^2)$ with respect to $\mu$, which is especially evident for high values of $\mu.$ 

Finally, we are interested in comparing the correctness of batching mechanisms against two natural benchmarks.
The first benchmark $V_\textsc{seq}$ is the sequential offering mechanism (see Section~\ref{section: seq_benchmark}) in the presence of strategic incentives.
As established in Theorem~\ref{theorem: main-result}, there exists a batching mechanism (in this case, $V_\textsc{greedy}^2$) which outperforms $V_\textsc{seq}$ for all $\mu < q$. 
However, observe that the number of batches is crucial. In contrast to $V_\textsc{greedy}^2$, $V_\textsc{greedy}^1$ achieves higher correctness than $V_\textsc{seq}$ only for $\mu < \frac{q}{2} + \frac{1}{4}$.

Comparing $V_\textsc{greedy}^2$ against $V_\textsc{seq}$, we observe that the performance gap  $c(V_\textsc{greedy}^2) - c(V_\textsc{seq})$ widens in the region $\mu \in [0, 1-q]$ as $\mu$ grows. 
This is because the object is always discarded for $\mu < q$ (see Figure~\ref{figure: sequential-offering}) which further explains why $c(V_{\textsc{seq}}) = 1 -\mu$ (Proposition~\ref{prop: sequential-two-agents}). For $\mu \in (1-q, q)$, the gap remains positive but decreases as $\mu$ increases, since the signals of the first two agents affect the outcome of $V_{\textsc{seq}}$ and thus decrease 
the chance of misallocation.
Note again that, due to Theorem~\ref{theorem: main-result}, for $\mu \geq q$, 
any batching mechanism $V \in \mathcal{V}$ achieves correctness equal to $c(V_{\textsc{seq}})$.

The second benchmark $V_\textsc{all}$ assumes that agents are not strategic (see Equation \eqref{eq: correctness_upper_bound}) and thus achieves the optimal correctness close to 1. The ratio  $c(V^J_{\textsc{greedy}})/ c(V_{\textsc{seq}})$, $J \in \{1,2\}$, serves as a measure of the \textit{price of anarchy}. 
We see that incentives put a significant cost on correctness, especially for higher values of $\mu \leq q$. Higher values of $q$ help decrease the price of anarchy. For any $q$, the maximum price of anarchy $c(V^J_{\textsc{greedy}})/ c(V_{\textsc{seq}})$ is observed around $\mu =q,$ and equals $1/q$ in a single-batch mechanism.

\section{Discussion and Conclusion}\label{sec:dis}

Empirical evidence suggests that herding is a contributing factor to the non-utilization of recovered organs. Our results indicate that herding can be reduced in organ allocation by using  batch offers and randomization when an organ is not allocated immediately.

Moreover, common organ quality metrics,  such as the Kidney Donor Risk Index (KDRI), do not accurately  reflect the organ's true quality when doctors have noisy signals over its quality (e.g., biopsy interpretations \citep{lentineVariationUseProcurement2019}). Our findings suggest that batch offers can be  a useful instrument in accurately updating organ quality and improving allocation.

In practice, batch offers can be an effective strategy for provisional offers; these are preliminary offers that are  utilized in the US to elicit whether doctors   would give serious consideration to the organ. Such changes help to better balance the need for learning about the organ's quality and respecting patients' priorities. 

Observe that our  mechanism may not always result in a Pareto improvement as patients with a high priority may be skipped due to the randomness in allocation. However, since patients' priorities depend on the organ quality, an organ allocation policy that is based on (real-time) assessed quality may  naturally alter the order of offers.

It is  worth noting that introducing batching offers into the allocation policy is aligned with expediting the organ offering process; this is especially  important in placing marginal organs in order to prevent discard due to the accumulation of  cold ischemia time.

We acknowledge that our model makes certain assumptions, such as behaviorally identical agents with a shared prior, utility, and information structure. In the context of kidney transplantation, the shared prior assumption can be justified by the use of the Kidney Donor Risk Index (KDRI), which ranks kidneys based on their relative risk. While the assumptions of identical utilities and information structures may not fully reflect the real organ allocation markets, the seminal work of \citet{smith2000pathological}  suggests that the qualitative essence of our findings remains valid in broader contexts, as long as there are no agents with unbounded signals. 
However, we recognize that the model does not capture all the complexities of real-world organ allocation, such as heterogeneous patient characteristics and varying utilities from the same organ. 

Overall, 
our findings shed light on how to reduce on herding and potentially the discard rate in organ allocation. 
Future research could explore the practical implementation of such  policy amendments and investigate their effects on the efficiency and equity of organ allocation.

\appendix

 \section{Optimal Solution in the Absence of Strategic Incentives}
 \label{appendix: no-incentives}
 
 \begin{lemma}
 	\label{lemma: optimal_threshold_no_incentives}
 	Suppose that agents are not strategic and voluntarily reveal their private signal value to the planner. Let
 	\begin{equation*}
 		\underbar{y} = 0.5 \, \frac{\log \left(\frac{1-\mu}{\mu}\right)}{\log \left(\frac{q}{1-q}\right)} +0.5 \, I.
 	\end{equation*}
 	Then, the planner achieves optimal correctness when she allocates the object if and only if $\sum_{i=1}^I \mathbf{1}\{s_i=g\} \geq \underbar{y}$.
 \end{lemma}
 
 \begin{proof}{\textit{Proof.}}
 	Conditional on $\omega = G$, the number $X_G$ of positive signals $s=g$ follows a binomial distribution, i.e., $X_G \sim \textrm{Binomial}(I, q)$. Conditional on $\omega = G$, the number $X_B $ of positive signals $s=g$ follows a binomial distribution, i.e., $X_B \sim \textrm{Binomial}(I, 1-q)$.
 	
 	Let $\underbar{y}$ denote the number of minimal positive signals required to allocate the object. For any number of positive signals higher than $\underbar{y}$, it is also optimal to allocate the object since the posterior belief that $\omega = G$ strictly increases.
 	
 	The threshold $\underbar{y}$ is the smallest integer that satisfies 
 	\begin{equation*}
 		\begin{split}
 			& \bbP(\omega = G \mid X_G = \underbar{y}) \geq 0.5 \\
 			\iff & \frac{\mu {I \choose  \underbar{y}} q^{\underbar{y}} (1-q)^{I- \underbar{y}} }{ \mu {I \choose  \underbar{y}} q^{\underbar{y}} (1-q)^{I- \underbar{y}} + (1-\mu) {I \choose  \underbar{y}} q^{I-\underbar{y}} (1-q)^{\underbar{y}} }\geq 0.5\\
 			\iff & \mu q^{\underbar{y}} (1-q)^{I- \underbar{y}} \geq (1-\mu) q^{I-\underbar{y}} (1-q)^{\underbar{y}}\\
 			\iff & \mu q^{2\underbar{y}-I} \geq (1-\mu) (1-q)^{2\underbar{y}-I}\\
 			\iff & \left( \frac{q}{1-q}\right)^{2\underbar{y}-I} \geq \frac{1-\mu}{\mu}.
 		\end{split}
 	\end{equation*}
 	
 	Recall that $q \in (0.5,1)$.
 	Taking the logarithm of the previous relation and rearranging the terms, we finally get that
 	\begin{equation*}
 		\underbar{y} = 0.5 \, \frac{\log \left(\frac{1-\mu}{\mu}\right)}{\log \left(\frac{q}{1-q}\right)} +0.5 \, I.
 	\end{equation*}
 \end{proof}
 
 \section{Omitted Analysis of the Sequential Offering Mechanism (Section~\ref{section: seq_benchmark})}

 \label{appendix: agents-sequential-offering}
 We first analyze the optimal strategies of the first three agents in the queue. Before taking any action, each agent observes his position in the queue in addition to his private signal and the common object prior $\mu$. In particular, by observing his position $i$, agent $i$ recognizes that the object is being offered to him because all of the $i-1$ preceding agents have rejected it. This observation contributes to his posterior belief about the object quality.

 	\textbf{Agent $1$.} 
 	Agent $1$ updates his posterior belief using Bayes' law based only on his  signal $s_1.$
 	\begin{align*}
 		\bbP(\omega=G | s_1) = \begin{cases}
 			\frac{\mu q}{\mu q + (1-\mu)(1-q)} &\text{ for } s_1 = g\\
 			\frac{\mu (1-q)}{\mu (1-q) + (1-\mu)q} &\text{ for } s_1 = b.
 		\end{cases}
 	\end{align*}
 	Since he receives utility $1$ if $\omega=G$ and $-1$ if $\omega=B,$ his expected utility to opt in is ${\bbP(\omega=G | s_1) - \left(1-\bbP(\omega=G | s_1)\right).}$ On the other hand, his utility to opt out is $0.$ We observe that agent $1$ 
 	is better off opting in if
 	$$\bbP(\omega=G | s_1) - \left(1-\bbP(\omega=G | s_1)\right) > 0, $$ which corresponds to
 	\begin{align*}
 		\begin{cases}
 			\mu q - (1-\mu)(1-q) >0 &\text{ when } s_1=g,\\
 			\mu(1-q) - (1-\mu)q > 0 &\text{ when } s_1=b.
 		\end{cases}.
 	\end{align*}
 	
 	Let $\alpha_1$ be agent $1$'s optimal action in this mechanism.
 	Then the following lemma characterizes $\alpha_1$.
 	
 	\begin{lemma}
 		Under the sequential offering mechanism, agent $1$ chooses action
 		\begin{align*} 
 			\alpha_1 = \begin{cases}
 				y &\text{ for } \mu > q\\
 				\bbI\{s_1=g\} &\text{ for } \mu \in (1-q,q]\\
 				n &\text{ for } \mu \le 1-q.
 			\end{cases}
 		\end{align*}
 		\label{claim: sequential-agent-1}
 	\end{lemma}
 	Lemma~\ref{claim: sequential-agent-1} shows that the sequential offering mechanism drives the first agent to follow his signal if it is more informative than the common prior (i.e., when $\mu \in (1-q,q]$). In the opposite case, where the prior is more informative (i.e., when either $\mu > q$ or $\mu \le 1-q$), the agent ignores his private signal. In particular, for high priors $\mu > q$ the object is always allocated to agent $1$ regardless of its true quality $\omega$.

 	\textbf{Agent $2$.}
 	When the second agent is offered the object, he knows that the first agent has already opted out. At the same time, agent 2 is aware of agent 1's optimal strategy as described in Lemma~\ref{claim: sequential-agent-1}. Therefore, if the prior satisfies $\mu \in (1-q,q],$ then agent 2 infers that agent 1 has followed his own signal $s_1=b.$ As such, for $\mu \in (1-q,q]$, agent 2 updates his posterior belief informed by this observation as follows:
 	\begin{align*}
 		\bbP(\omega=G | s_2, \alpha_1=n) = \begin{cases}
 			\mu &\text{ for } \mu \in (1-q,q], s_2 = g\\
 			\frac{\mu (1-q)^2}{\mu (1-q)^2 + (1-\mu)q^2} &\text{ for } \mu \in (1-q,q], s_2 = b.
 		\end{cases}
 	\end{align*}
 	On the other hand, for $\mu \notin (1-q,q]$, it holds that
 	\begin{align*}
 		\bbP(\omega=G | s_2, \alpha_1=n) = \begin{cases}
 			\frac{\mu q}{\mu q + (1-\mu)(1-q)} &\text{ for } \mu \notin (1-q,q], s_2 = g\\
 			\frac{\mu (1-q)}{\mu (1-q) + (1-\mu)q} &\text{ for } \mu \notin (1-q,q], s_2 = b,
 		\end{cases}
 	\end{align*}
 	in which case agent 1's action is uninformative to agent 2. 
 	
 	\begin{lemma}
 		Under the sequential offering mechanism, agent $2$ chooses\footnote{In the sequential offering mechanism, if $\mu > q$ then the object is never offered to agent $i \ge 2,$ because agent 1 always opts in.}
 		\begin{align*} 
 			\alpha_2 = \begin{cases}
 				\bbI\{s_2=g\} &\text{ for } \mu \in (1/2,q]\\
 				n &\text{ for } \mu \leq 1/2.
 			\end{cases}
 		\end{align*}
 		\label{claim: sequential-agent-2}
 	\end{lemma}
 	\begin{proof}{\textit{Proof.}}
 		Similarly to agent 1, agent 2 prefers to opt in if 
 		\begin{align}
 			\bbP(\omega=G|s_2,\alpha_1=n) - (1-\bbP(\omega=G|s_2,\alpha_1=n) ) > 0.
 			\label{eqn: sequential-agent-2}
 		\end{align}
 		For $\mu \in (1-q,q]$ and $s_2=g,$ Equation \eqref{eqn: sequential-agent-2} holds when $\mu > 1/2.$ For $\mu \in (1-q,q]$ and $s_2=b,$ Equation \eqref{eqn: sequential-agent-2} corresponds to 
 		\begin{align*}
 			\mu(1-q)^2-(1-\mu)q^2 > 0.
 		\end{align*}
 		However, for any $\mu \in (1-q,q]$ we have 
 		\begin{align*}
 			\mu(1-q)^2-(1-\mu)q^2  \le q(1-q)^2-(1-q)q^2 
 			= q(1-q)\left(1-2q \right)
 			< 0,
 		\end{align*}
 		which makes Equation \eqref{eqn: sequential-agent-2} infeasible.
 		
 		For $\mu \le 1-q,$ notice that 
 		\begin{align*}
 			\mu(1-q) - (1-\mu)q < \mu q - (1-\mu)(1-q) \le 0. 
 		\end{align*}
 		Thus Equation \eqref{eqn: sequential-agent-2} is again infeasible regardless of $s_2.$
 		
 		For $\mu=1/2$ exactly, by our assumption, he simply rejects the object.
 		Therefore, agent 2 prefers to opt in if $\mu \in (1/2,q]$ and $s_2=g.$ Otherwise, he prefers to opt out.
 	\end{proof}
 	
 	Lemma~\ref{claim: sequential-agent-2} shows that agent 2 follows his signal if $\mu > 1/2,$ whereas he ignores the signal and opts out if $\mu \leq 1/2.$ The characterizations of $\alpha_1$ and $\alpha_2$ offer implications that will be useful to analyze the optimal strategies of subsequent agents. First, for $\mu > 1/2$, 
 	Lemma~\ref{claim: sequential-agent-1} and Lemma~\ref{claim: sequential-agent-2} together suggest that both agent 1 and agent 2 follow their signals. Therefore, the availability of the object after being offered to agent 2 implies to subsequent agents that both ${s_1=b}$ and ${s_2=b.}$ Second, for $\mu \leq 1/2,$ agent 2's action is uninformative about $s_2$. Agent 1's action, however, can be informative depending on the value of $\mu.$ If $\mu \in (1-q, 1/2)$ then the subsequent agents can infer that $s_1=b$ because agent 1 follows his signal. If $\mu \le 1-q,$ however, agent 1's action is also uninformative. 
 	
 	Notice here that, whenever agent $1$'s or $2$'s opt-out action is informative, his private signal (that the subsequent agents can perfectly infer) is $b.$ We use this implication for our analysis next.

 	\textbf{Agent $3$.}
 	If the object is offered to agent 3, one of the three events must have occurred: (i) $\mu \in (1/2,q]$ \emph{and} $s_1=s_2=b$, or (ii) $\mu \in (1-q, 1/2]$ \emph{and} $s_1=b$, or (iii) $\mu \le 1-q.$ Nonetheless, we claim that in any case,
 	agent 3 would always choose to opt out.
 	
 	\begin{lemma}
 		Under the sequential offering mechanism, agent 3 always chooses to opt out: $\alpha_3 = n.$
 		\label{claim: sequential-agent-3}
 	\end{lemma}
 	\begin{proof}{\textit{Proof.}}
 		Agent 3 will prefer to opt in if 
 		\begin{align}
 			\bbP(\omega=G|s_3,\alpha_1=\alpha_2=n) - (1-\bbP(\omega=G|s_3,\alpha_1=\alpha_2=n)) > 0.
 			\label{eqn: sequential-agent-3}
 		\end{align}
 		For $\mu \in (1/2,q],$ agent 3's posterior belief given her signal is 
 		\begin{align*}
 			\bbP(\omega=G|s_3,\alpha_1=\alpha_2=n) &= \bbP(\omega=G|s_3,s_1=s_2=b)\\&= \begin{cases}
 				\frac{\mu(1-q)^2q}{\mu(1-q)^2 q +(1-\mu)q^2(1-q)} &\text{ for } \mu \in (1/2,q], s_3=g\\
 				\frac{\mu(1-q)^3}{\mu(1-q)^3+(1-\mu)q^3} &\text{ for } \mu \in (1/2,q], s_3=b
 			\end{cases}.
 		\end{align*}
 		
 		For $s_3 = g$, we have 
 		\begin{align*}
 			\mu(1-q)^2 q -(1-\mu)q^2 (1-q) = (1-q) q [\mu(1-q) - (1-\mu)q] \le 0.
 		\end{align*}
 		The proof is analogous for $s_2 = b$.
 		As a result, for $\mu \in (1/2,q]$ agent 3  prefers to opt out regardless of his own signal. 
 		
   	For $\mu \in (1-q, 1/2],$ agent 3 infers that $s_1 =b.$ Agent 3's posterior belief equals 
 		\begin{align*}
 			\bbP(\omega=G|s_3,\alpha_1=\alpha_2=n) & = \bbP(\omega=G|s_3,s_1=b) \\
 			&= \begin{cases}
 				\mu &\text{ for } \mu \in (1-q, 1/2], s_3=g\\
 				\frac{\mu (1-q)^2}{\mu (1-q)^2 + (1-\mu)q^2} &\text{ for } \mu \in (1-q, 1/2], s_3=b.  
 			\end{cases}
 		\end{align*}
 		Then Equation \eqref{eqn: sequential-agent-3} is infeasible because $$\mu - (1-\mu)  =2\mu -1 \leq 0.$$
 		
 		Finally, for $\mu \le 1-q,$
 		\begin{align*}
 			\bbP(\omega=G|s_3,\alpha_1=\alpha_2=n) &= \bbP(\omega=G|s_3) \\&= \begin{cases}
 				\frac{\mu q}{\mu q + (1-\mu)(1-q)} &\text{ for } \mu \le 1-q, s_3=g\\
 				\frac{\mu (1-q)}{\mu (1-q) + (1-\mu)q} &\text{ for } \mu \le 1-q, s_3=b  
 			\end{cases}
 		\end{align*}
 		where we have
 		\begin{align*}
 			\mu (1-q) - (1-\mu)q < \mu q - (1-\mu)(1-q) 
 			= \mu + q - 1
 			\le 0,
 		\end{align*}
 		making Equation \eqref{eqn: sequential-agent-3} infeasible.
 	\end{proof}
 	
 	
 	Now we extend Lemma~\ref{claim: sequential-agent-3} to any remaining agent in the queue. Let $\alpha_i$ denote the optimal action of agent $i.$
 	\begin{lemma}
 		In the sequential offering mechanism, any agent other than the first two always opts-out. That is, for any $i \ge 3:$
 		$$\alpha_i = n. $$
 		\label{lemma: sequential-after-3}
 	\end{lemma}
 	\begin{proof}{\textit{Proof.}}
 		We use mathematical induction on $i$ to prove the statement. For the basis $i=3$, we have that due to Lemma~\ref{claim: sequential-agent-3} agent 3 always opts out thus his action is always uninformative to the subsequent agents. 
 		Next, as the inductive step, consider some agent $i.$ Suppose that agent $i$ updates his posterior belief in some manner such that it is optimal for him to opt out regardless of $s_i.$ Because this action is uninformative, the next agent $i+1$ must update his posterior belief the same way as agent $i.$ Therefore, it is also optimal for agent $i+1$ to opt out.
 	\end{proof}
 	Another way to interpret Lemma~\ref{lemma: sequential-after-3} is that if neither agent 1 nor agent 2 opts in, then the object will be discarded. Using this result along with Lemma~\ref{claim: sequential-agent-1} and Lemma~\ref{claim: sequential-agent-2} we characterize the outcome of the sequential offering mechanism (Lemma~\ref{lemma: sequential-allocation} in Section~\ref{section: seq_benchmark}): 
 	
 	\lemmasequentialallocation*
  
 	\begin{proof}{\textit{Proof.}}
 		The proof follows directly from combining Lemmas~\ref{claim: sequential-agent-1} to \ref{lemma: sequential-after-3}.
 	\end{proof}

 	\propositionseqcorrect*

 	\begin{proof}{\textit{Proof.}}
 		Part (i) directly follows from Lemma~\ref{lemma: sequential-allocation}. For part (ii), notice that for $\mu \in (1/2,q],$ the object is allocated if either $s_1=g$ or $s_2=g.$ Conditional on $\omega=G,$ this happens with probability $1-(1-q)^2.$ The object is discarded if $s_1=s_2=b,$ which happens with $q^2$ conditional on $\omega=B.$ Recall that $\bbP(\omega=G) = \mu.$ Therefore the correctness of the outcome is $\mu\left( 1-(1-q)^2\right) + (1-\mu)q^2 = 2\mu q(1-q)+q^2.$
 		
 		For $\mu \in (1-q,1/2]$ the object is allocated if  $s_1=g.$ Conditional on $\omega=G,$ the object is allocated with probability $q.$ The object is discarded when $s_1=b$ and conditional on $\omega=B$ this happens with probability $q.$ Hence the outcome is correct with probability $q.$
 	\end{proof}

 	\section{Proofs for Section \ref{section: voting-mechanisms}}
 	\label{appendix: single_batch_proofs}
 	\subsection{Proofs for Section \ref{section: single-batch}}
 	 
     \lemmaiccombined*
 	\begin{proof}{\textit{Proof.}}
 		Proof of (i):
 		Consider a single-batch batching mechanism $V_{\{K\}}.$ Suppose that agent $i$ chooses to opt in. Let $\mathcal{G}_K$ be the probability that the object gets allocated to agent $i$ conditional on the object quality being good. Similarly, let $\mathcal{B}_K$ be the probability that the object gets allocated to agent $i$ conditional on the object quality being bad. Then $\mathcal{G}_K$ and $\mathcal{B}_K$ can be computed as
 		\begin{align*}
 			\mathcal{G}_K &= \sum_{y=\frac{K+1}{2}}^K \frac{1}{y} {K-1 \choose y-1} q^{y-1}(1-q)^{K-y}\\
 			\mathcal{B}_K &= \sum_{y=\frac{K+1}{2}}^K \frac{1}{y} {K-1 \choose y-1} q^{K-y}(1-q)^{y-1}.
 		\end{align*}
 		Given private signal $s_i,$ his expected utility conditional on each action is
 		\begin{align*}
 			u_i(y;s_i) &= \bbP(\omega=G | s_i) \mathcal{G}_K - \bbP(\omega=B | s_i=g) \mathcal{B}_K \\
 			u_i(n;s_i) &= 0.
 		\end{align*}
 		Then the mechanism $V_{\{K\}}$ is incentive-compatible if $u_i(y;g) \ge u_i(n;g)$ and $u_i(y;b) \le u_i(n;b).$ That is, $V_{\{K\}}$ is incentive-compatible for prior $\mu$ if
 		\begin{align}
 			u_i(y;g) \ge u_i(n;g) &\iff \bbP(\omega=G | s_i=g) \mathcal{G}_K - \bbP(\omega=B | s_i=g) \mathcal{B}_K \ge 0 \nonumber
 			\\
 			&\iff
 			\frac{\bbP(\omega=G) \bbP(s_i=g|\omega=G)  \mathcal{G}_K}{\bbP(s_i=g)} - \frac{\bbP(\omega=B) \bbP(s_i=g|\omega=B) \mathcal{B}_K}{\bbP(s_i=g)} \ge 0  \nonumber \\
 			&\iff
 			\mu q \mathcal{G}_K - (1-\mu)(1-q)\mathcal{B}_K \ge 0
 			\label{eqn: IC-good-signal}
 		\end{align}
 		and 
 		\begin{align}
 			u_i(y;b) \ge u_i(n;b) &\iff \bbP(\omega=G | s_i=b) \mathcal{G}_K - \bbP(\omega=B | s_i=b) \mathcal{B}_K \le 0 \nonumber
 			\\
 			&\iff
 			\frac{\bbP(\omega=G) \bbP(s_i=b|\omega=G)  \mathcal{G}_K}{\bbP(s_i=b)} - \frac{\bbP(\omega=B) \bbP(s_i=b|\omega=B) \mathcal{B}_K}{\bbP(s_i=b)} \le 0  \nonumber\\
 			&\iff
 			\mu (1-q) \mathcal{G}_K - (1-\mu)q \mathcal{B}_K \le 0.
 			\label{eqn: IC-bad-signal}
 		\end{align}
 		Given that the left-hand sides of both \eqref{eqn: IC-good-signal} and \eqref{eqn: IC-bad-signal} monotonically increases with $\mu,$ there should be some thresholds of prior, namely $\underline{\mu}_K \in (0,1)$ and $\overline{\mu}_K \in (0,1)$, such that
 		$$\mu q \mathcal{G}_K - (1-\mu)(1-q)\mathcal{B}_K \ge 0 \quad \forall \mu \ge \underline{\mu}_K$$ 
 		$$\mu (1-q) \mathcal{G}_K - (1-\mu)q \mathcal{B}_K \le 0 \quad \forall \mu \ge \overline{\mu}_K,$$
 		where the thresholds solve the following indifference conditions:
 		\begin{align}
 			\frac{\underline{\mu}_K}{1-\underline{\mu}_K} &= \frac{1-q}{q} \frac{\mathcal{B}_K}{\mathcal{G}_K} 
 			\label{eqn: under-mu}\\
 			\frac{\overline{\mu}_K}{1-\overline{\mu}_K} &= \frac{q}{1-q} \frac{\mathcal{B}_K}{\mathcal{G}_K}.
 			\label{eqn: over-mu}
 		\end{align}
 		Because $q > \frac{1}{2}$ it naturally follows that $\underline{\mu}_K < \overline{\mu}_K.$ Therefore, there exists a threshold policy such that any given $V_{\{K\}}$ is incentive-compatible for a prior $\mu$ that satisfies $\underline{\mu}_K \le \mu \le \overline{\mu}_K.$ 
 			
 		Next we fully characterize $\underline{\mu}_K$ and $\overline{\mu}_K.$ To do so, we use Lemma~\ref{property: combinatorial} to write
 		$$\mathcal{G}_K = \frac{1}{q} \cdot \frac{1}{K} \sum_{y=\frac{K+1}{2}}^{K} {K \choose y} q^y (1-q)^{K-y} $$
 		$$\mathcal{B}_K = \frac{1}{1-q} \cdot \frac{1}{K} \sum_{y=0}^{\frac{K+1}{2}-1} {K \choose y} q^y (1-q)^{K-y}$$
 		and therefore 
 		\begin{align}
 			\frac{\mathcal{B}_K}{\mathcal{G}_K} = \frac{q}{1-q} \cdot \frac{1-\bbP\left(X_K \ge \frac{K+1}{2}\right)}{\bbP\left(X_K \ge \frac{K+1}{2}\right)}
 			\label{eqn: B/G-simple}
 		\end{align}
 		where $X_K$ is a $\text{binomial}(K,q)$ random variable. Plugging this into \eqref{eqn: under-mu} and \eqref{eqn: over-mu} we obtain 
 		\begin{align}
 			\underline{\mu}_K &= 1 - \bbP\left(X \ge \frac{K+1}{2}\right) \label{eqn: majority-lower-mu-tail}\\
 			\overline{\mu}_K &= \frac{q^2 \left( 1 - \bbP\left(X \ge \frac{K+1}{2}\right)\right)}{q^2 \left( 1 - \bbP\left(X \ge \frac{K+1}{2}\right)\right) + (1-q)^2 \bbP\left(X \ge \frac{K+1}{2}\right)} \label{eqn: majority-upper-mu-tail}\\
 			&= \frac{q^2 \underline{\mu}_K}{q^2 \underline{\mu}_K + (1-q)^2 (1-\underline{\mu}_K)}.\nonumber
 		\end{align}
 		Letting $\mathcal{I}_K \triangleq \left( \underline{\mu}_K, \overline{\mu}_K \right)$ concludes the proof.

 		Proof of (ii): 
 		By Lemmas \ref{lemma: ic-interval-decrease} and \ref{lemma: ic-intervals-overlap} (stated below), the set of $K$ such that $\mu \in \mathcal{I}_K$  must be comprised of consecutive odd numbers. Therefore, for any $K$ that satisfies $\underline{K}(\mu) \le K \le \overline{K}(\mu),$ the corresponding $V_{\{K\}}$ must be incentive-compatible. 
 	\end{proof}

 	We show two technical lemmas with regards to $\mathcal{I}_K$ that were used to prove Lemma~\ref{lemma: ic-combined}. In doing so, we make use of Lemma~\ref{property: tail-increase}.
 	
 	\begin{restatable}[Decreasing $\mathcal{I}_K$]{lemma}{lemmaintervaldecrease}
 		$\mathcal{I}_K$ decreases with $K.$ That is, the endpoints $\overline{\mu}_K$ and $\underline{\mu}_K$ both decrease with $K.$ In particular, $\overline{\mu}_1=q$ and $\lim_{K \rightarrow \infty} \underline{\mu}_K = 0.$
 		\label{lemma: ic-interval-decrease}
 	\end{restatable}
 	\begin{proof}{\textit{Proof.}}
 		Recall that the endpoints $\underline{\mu}_K$ and $\overline{\mu}_K$ are computed as Equations \eqref{eqn: majority-lower-mu-tail}, both of which by Lemma~\ref{property: tail-increase}, decrease with $K.$ To see the second part of the statement, since $$\bbP\left( X_1 \ge 1 \right) = \bbP\left(X_1 = 1\right) = q$$ it follows that
 		$\overline{\mu}_1 = \frac{q^2(1-q)}{q^2(1-q) + (1-q)^2 q} = q.$ Moreover, 
 		\begin{align*}
 			\lim_{K \rightarrow \infty} \underline{\mu}_K = \lim_{K\rightarrow\infty} \frac{q^2 \left(1 - \bbP\left(X_K \ge \frac{K+1}{2}\right)\right)}{q^2 \left(1 - \bbP\left(X_K \ge \frac{K+1}{2}\right)\right) + (1-q)^2 \bbP\left(X_K \ge \frac{K+1}{2}\right)} 
 			= 0
 		\end{align*}
 		where the last equality is due to $\lim_{K \rightarrow \infty} \bbP\left(X_K \ge \frac{K+1}{2}\right) = 1$ by Lemma~\ref{property: tail-increase}.
 	\end{proof}
 	
 	\begin{restatable}[Overlapping $\mathcal{I}_K$]{lemma}{lemmaintervalnoskipping}
 		For any batch size $K \ge 3,$ the interval $\mathcal{I}_K$ intersects with $\mathcal{I}_{K-2}.$ That is, 
 		$$\underline{\mu}_K < \underline{\mu}_{K-2} < \overline{\mu}_K < \overline{\mu}_{K-2}.$$
 		\label{lemma: ic-intervals-overlap}
 	\end{restatable}
 	\begin{proof} {\textit{Proof sketch.}} 
 		The first and third inequalities are immediate from Lemma \ref{lemma: ic-interval-decrease}. Showing the second inequality involves finding appropriate bounds for $\frac{\mathcal{B}_K}{\mathcal{G}_K} \frac{\mathcal{G}_{K-2}}{\mathcal{B}_{K-2}}.$ Due to the length of the proof, we defer this part to Appendix~\ref{appendix: overlap-proof}.
 	\end{proof}
 	
 	Figure \ref{fig: feasible-priors} illustrates Lemmas \ref{lemma: ic-interval-decrease} and \ref{lemma: ic-intervals-overlap}.

 	\lemmaexistencecorrectness* 
 	\begin{proof}{\textit{Proof.}}
 		We first show the existence of incentive-compatible $V_{\{K\}}$ for any $\mu < q.$ Consider some $\mathcal{K}.$ By Lemmas \ref{lemma: ic-interval-decrease} and \ref{lemma: ic-intervals-overlap}, 
 		$$\bigcup_{K=1}^{\mathcal{K}} \mathcal{I}_K = (\underline{\mu}_{\mathcal{K}},q)$$
 		which implies that that no value of prior $\mu$ between $\underline{\mu}_{\mathcal{K}}$ and $q$ is skipped. Using the asymptotic result from the second part of Lemma \ref{lemma: ic-interval-decrease},
 		$$\lim_{{\mathcal{K}} \rightarrow \infty} \bigcup_{K=1}^{{\mathcal{K}}} \mathcal{I}_K = \left(\lim_{{\mathcal{K}} \rightarrow \infty} \underline{\mu}_{\mathcal{K}},q\right) = (0,q).$$
 		Therefore for any $\mu$ such that $\mu \le q,$ there exists at least one $K$ such that $\mu \in \mathcal{I}_K.$

 		For the second part, we have
 		\begin{align*}
 			c\left(V_{\{K\}}\right) &= \mu \, \bbP_{V_{\{K\}}}( Z = 1 \mid \w = G) + (1-\mu) \, \bbP_{V_{\{K\}}} ( Z = 0  \mid \w = B) \\
 			&= \mu \bbP\left(Y \ge \frac{K+1}{2} \ \middle|\  \w = G \right) + (1-\mu)\bbP\left(Y < \frac{K+1}{2} \ \middle|\  \w = B\right)\\
 			&= \bbP\left(X_K \ge \frac{K+1}{2}\right).
 		\end{align*}
 	\end{proof}

 	\lemmaoptimalsinglebatch*
 	\begin{proof}{\textit{Proof.}}
 		By Lemma~\ref{lemma: existence-and-correctness},  $c\left(V_{\{K\}}\right)= \bbP\left(X_K \ge \frac{K+1}{2}\right)$, which by Lemma~\ref{property: tail-increase}, increases with $K.$
 	\end{proof}
 	
 	\lemmacomparative*
 	\begin{proof}{\textit{Proof.}}
 		See the proof of Lemma \ref{lemma: ic-interval-decrease}.
 	\end{proof}

 	\subsection{Proofs for Section \ref{section: multi-batch}}
 	\label{appendix: proofs-multi-batch}
 	\lemmaicgreedy*
 	\begin{proof}{\textit{Proof.}}
 		Each agent participates in at most one batch. Furthermore, agents in the same batch $j$ share the same belief $\mu_{j-1}$ with the planner. Thus, given the current belief $\mu_{j-1}$ and batch size $K_j$ (given the outcome of batch $j-1$), the mechanism $V_{\{\pi_j\}_{j=1}^\infty}$ restricted to batch $j$ is equivalent to a single-batch mechanism $V_{\{K_j\}}$ with a current belief $\mu_{j-1}$.
 		Therefore, by induction, the mechanism $V_{\{\pi_j\}_{j=1}^\infty}$  is incentive-compatible if and only if for each batch $j,$ the single-batch mechanism $V_{\{K_j\}}$ with belief $\mu_{j-1}$ is incentive-compatible. 
 	\end{proof}

 	\begin{restatable}{lemma}{corollaryimpossibility}
 		For $\mu \geq q,$ there is no incentive-compatible 
 		batching mechanism; under any $V \in \mathcal{V},$ every agent opts in. 
 		\label{corollary: impossibility}
 	\end{restatable}
 	\begin{proof}{\textit{Proof.}}
 		By Lemma~\ref{lemma: IC-greedy}, it suffices to consider single-batch  mechanisms. For any $K,$ Lemma \ref{lemma: ic-interval-decrease} implies that $$\overline{\mu}_K \le \max_{K} \overline{\mu}_K = \overline{\mu}_1= q.$$ Therefore for all $\mu \ge q,$ there exists no $K$ such that $\mu < \overline{\mu}_K.$ For these priors, every agent will choose (possibly untruthfully) $\alpha_i=y$ under any batching mechanism, and therefore the mechanism always allocates the object (and terminates after one batch). In expectation, such an outcome is correct with probability $\mu.$ Notice that the correctness of this outcome is equivalent to that under $V_\textsc{seq}$ for $\mu \ge q.$
 	\end{proof}

 	\begin{lemma}
 		Suppose that the object is offered to batch $j$ but is not allocated. Then $\mu_j < \mu_{j-1}.$
 		\label{lemma: posterior-after-failure}
 	\end{lemma}
 	\begin{proof}{\textit{Proof.}}
 		Object is not allocated to batch $j$ if $Y_j \le \frac{K_j - 1}{2}.$ Using the posterior update rule in \eqref{eq: prior_mu_j},
 		\begin{align*}
 			\mu_j &= \frac{\mu_{j-1}q^{Y_j}(1-q)^{K_j-Y_j}}{\mu_{j-1}q^{Y_j}(1-q)^{K_j-Y_j} + (1-\mu_{j-1})q^{K_j-Y_j}(1-q)^{Y_j}} \\
 			&\le 
 			\frac{\mu_{j-1}(1-q)}{\mu_{j-1}(1-q)+ (1-\mu_{j-1})q} \\
 			&< \mu_{j-1}.
 		\end{align*}
 	\end{proof}
 	
 	\propjgreedy*
 	\begin{proof}{\textit{Proof.}}
 		Part (i) is by Lemma~\ref{lemma: IC-greedy}. 
 		
 		For part (ii), by Lemma~\ref{lemma: posterior-after-failure}, $\mu_j$ decreases with $j.$ Each batch size $K_j$ is chosen as ${K_j = \overline{K}(\mu_{j-1}),}$ which, by Lemma~\ref{lemma: comparative}, weakly increases with $j.$
 		
 		For part (iii), consider any integer $J$ and the greedy mechanism $V_\textsc{greedy}^J.$ Suppose that the object has not been allocated up to batch $J.$ Then should the planner offer the object to another batch?
 		
 		On the one hand, the expected correctness conditional on stopping at batch $J$ is
 		\begin{align*}
 			\bbP\left(\w=B \mid \mu_J \right)
 			&= 1 - \mu_J.
 		\end{align*}
 		
 		On the other hand, the expected correctness conditional on offering to an additional batch $J+1$ is equal to $c\left(V_{\{K_{J+1}\}}\right),$ which is the correctness of a single-batch batching mechanism  where $K_{J+1} = \overline{K}(\mu_J)$ is the size of batch $J+1$ chosen in a greedy manner.  Then by Lemma~\ref{lemma: existence-and-correctness}, the expected correctness is 
 		\begin{align*}
 			\bbP\left(X_{\overline{K}(\mu_J)} \ge \frac{\overline{K}(\mu_J)+1}{2} \right).
 		\end{align*}
 		Consider the interval $\mathcal{I}_{\overline{K}(\mu_J)}$ that defines the values of incentive-compatible  priors for the single-batch  mechanism $V_{\{\overline{K}(\mu_{J})\}}.$ In particular, consider the lower endpoint of this interval,
 		$\underline{\mu}_{\overline{K}(\mu_J)}.$ Then, because $\mu_J \in \mathcal{I}_{\overline{K}(\mu_J)}$ by construction, 
 		it must be that $\mu_J > \underline{\mu}_{\overline{K}(\mu_J)}.$ Furthermore, we can compute  $$\underline{\mu}_{\overline{K}(\mu_J)} = 1 - \bbP\left(X_{\overline{K}(\mu_J)} \ge \frac{\overline{K}(\mu_J)+1}{2} \right).$$
 		
 		Hence, the
 		expected correctness achieved by an additional batch is
 		\begin{align*}
 				c\left(V_{\{\overline{K}(\mu_J)\}}\right) &=
 			\bbP\left(X_{\overline{K}(\mu_J)} \ge \frac{\overline{K}(\mu_J)+1}{2} \right)\\
 			&= 1 - \underline{\mu}_{\overline{K}(\mu_J)} \\
 			&> 1 - \mu_J.
 		\end{align*}
 	\end{proof}
 	
 	\corollarygreedy*
 	\begin{proof}{\textit{Proof.}}
 		Immediate from Theorem~\ref{theorem: main-result} and Proposition~\ref{prop: greedy-j-properties} (part (iii)). 
 	\end{proof}

 	\subsection{Proof of Theorem~\ref{theorem: main-result}}
 	\label{appendix: main-theorem-proof}
 	\thmmainresult*
 	\begin{proof}{\textit{Proof.}}
 		In Proposition~\ref{prop: sequential-two-agents} we calculate the correctness of $V_{SEQ},$ which will be used as our benchmark. We split the proof into three different cases:
 		\begin{itemize}
 			\item Case I: $\mu < 1-q,$ so by Proposition~\ref{prop: sequential-two-agents}, $$c(V_\textsc{seq})=1-\mu.$$
 			By Lemmas~\ref{lemma: existence-and-correctness}, there must exist an incentive-compatible single-batch  mechanism $V_{\{K\}}$ such that $\mu \in \mathcal{I}_K.$ In particular, under such a mechanism, $$\mu > 1 - \bbP\left( K \ge \frac{K+1}{2} \right) = 1 - c\left(V_{\{K\}}\right).$$ Therefore,
 			$$ c(V_{\{K\}}) > 1-\mu = c(V_\textsc{seq}).$$
 			
 			\item Case II: $\mu \in [1-q,1/2).$
 			$$c(V_\textsc{seq}) = q,$$
 			and more importantly, for these values of prior, $V_\textsc{seq}$ is equivalent to $V_{\{1\}}$ that is incentive-compatible. Therefore, we use the observation (by Lemma~\ref{lemma: existence-and-correctness} and \ref{property: tail-increase}) that the correctness of an incentive-compatible $V_{\{K\}}$ increases with $K.$ 
 			
 			Then it is sufficient to show that there exists an incentive-compatible $V_{\{K\}}$ such that $K \ge 3.$ To this end, we show that $V_{\{3\}}$ is incentive-compatible. The interval $\mathcal{I}_3$ computes to  
 			\begin{align*}
 				\mathcal{I}_3 = \left((1-q)^2(2q+1), \frac{q}{2}+\frac{1}{4}\right)
 			\end{align*}
 			where $(1-q)^2(2q+1) < 1-q$ and $\frac{q}{2}+\frac{1}{4} > \frac{1}{2}.$ Hence $$[1-q,1/2) \subset \mathcal{I}_3,$$ implying that $V_{\{3\}}$ is incentive-compatible for any $\mu \in [1-q,1/2).$ 
 			
 			\item Case III: $\mu \in [1/2, q/2+1/4],$ in which case $\frac{1}{2} \le \mu < q.$ By Proposition~\ref{prop: sequential-two-agents},
 			$$c(V_\textsc{seq}) = 2\mu q(1-q)+q^2.$$ As in Case II, 
 			$$[1/2, q/2+1/4] \subset \mathcal{I}_3,$$ and therefore $V_{\{3\}}$ is incentive-compatible, which achieves (by Lemma~\ref{lemma: existence-and-correctness})
 			$$c\left(V_{\{3\}}\right) = q^3 + 3q^2(1-q) > 2\mu q(1-q)+q^2 = c(V_\textsc{seq}).$$

 			\item Case IV: $\mu \in (q/2 + 1/4, q).$ 
 			$$c(V_\textsc{seq}) = 2\mu q(1-q) + q^2.$$ Notice that $V_{\{3\}}$ is no longer incentive-compatible for these values of $\mu.$ Since $\mathcal{I}_1 = (1-q,q),$ $V_{\{1\}}$ is the only incentive-compatible single-batch  mechanism, which however achieves $c\left(V_{\{1\}}\right) = q < c(V_\textsc{seq}).$ 
 			
 			Motivated by Proposition~\ref{prop: greedy-j-properties} (especially, part (iii)), we turn to batching mechanisms with multiple batches instead. Consider $V_\textsc{greedy}^2.$ Since $V_{\{1\}}$ is the only incentive-compatible single-batch  mechanism, in order for $V_\textsc{greedy}^2$ to be incentive-compatible, it must have $K_1=1.$
 			The correctness of this mechanism is then
 			\begin{align*}
 				c(V_\textsc{greedy}^2) &= \bbE_{V_\textsc{greedy}^2} \left[\bbI(\w=G \cap Z=1) + \bbI(\w=B \cap Z=0) \right] \\
 				&= \bbP\left(Y_1=1 \cap \w=G \mid K_1=1 \right)\\
 				&\qquad + \bbP\left(Y_1=0 \mid K_1=1 \right) \, \max_{K_2} \bbE \left[c(V_{\{K_2\}}) \mid Y_1=0, K_1=0 \right] \\
 				&= \mu q + \left(\mu (1-q) + (1-\mu)q \right) \, \bbE \left[c(V_{\{\overline{K}(\mu_1)\}}) \mid \mu_1 = \frac{\mu(1-q)}{\mu(1-q)+(1-\mu)q} \right].
 			\end{align*}
 			
 			Here, in the event where the object is not allocated to the first batch (i.e., $Y_1=0$), $$\mu_1 \in \left(\frac{\left(\frac{q}{2}+\frac{1}{4}\right)(1-q)}{\left(\frac{q}{2}+\frac{1}{4}\right)(1-q) + \left(1-\frac{q}{2}-\frac{1}{4}\right)q}, \frac{1}{2} \right)$$ each of which is achieved by plugging in $\mu = \frac{q}{2} + \frac{1}{4}$ and $\mu = q$ respectively. In particular, observe that $$\mu_1 < \frac{1}{2} < \frac{q}{2} + \frac{1}{4} = \overline{\mu}_3,$$ where $\overline{\mu}_3$ denotes the upper endpoint of the interval $\mathcal{I}_3.$ This in turn implies 
 			\begin{align*}
 				\overline{K}(\mu_1) \ge \overline{K}(\overline{\mu}_3) = 3
 			\end{align*}
 			where the inequality is by Lemma~\ref{lemma: comparative}. Therefore,
 			\begin{align*}
 				\bbE \left[c(V_{\{\overline{K}(\mu_1)\}}) \mid \mu_1 = \frac{\mu(1-q)}{\mu(1-q)+(1-\mu)q} \right] \ge c(V_{\{3\}}) = q^3 + 3q^2(1-q).
 			\end{align*}
 			Using this observation,
 			\begin{align*}
 				c(V_\textsc{greedy}^2) 
 				\ge \mu q + \left(\mu (1-q) + (1-\mu)q \right) \left(q^3 + 3q^2(1-q) \right)
 			\end{align*}
 			and more importantly,
 			\begin{align*}
 				c(V_\textsc{greedy}^2) - c(V_\textsc{seq}^2) &\ge \mu q + \left(\mu (1-q) + (1-\mu)q \right) \left(q^3 + 3q^2(1-q) \right) - 2\mu q (1-q) - q^2.
 			\end{align*}

 			The partial derivative (with respect to $\mu$) of the inequality's right-hand side is negative. Therefore, plugging in $\mu = q$ yields
 			\begin{align*}
 				c(V_\textsc{greedy}^2) - c(V_\textsc{seq}^2) &> q^2 + 2q (1-q)  \left(q^3 + 3q^2(1-q) \right) - 2q^2 (1-q) - q^2 
 			\end{align*}
 			where the right-hand side term is always positive for any $q \in (0.5,1).$
 			As a result, $$c(V_\textsc{greedy}^2) > c(V_\textsc{seq}).$$

 			Notice in Cases II and III, we use $V_{\{3\}}$ as a sufficient condition. However, one can always further improve correctness by using either the optimal $V_\textsc{greedy}^1$ for the given prior or  any other $V_\textsc{greedy}^J$ (including $V_\textsc{greedy}$ implemented as in Algorithm~\ref{alg: greedy}) with more batches.  
 			\end{itemize}  
 			
 			Finally, the proof of the second statement is given in Lemma~\ref{corollary: impossibility}.
 			\end{proof}

 	\subsection{Remaining Proof of Lemma~\ref{lemma: ic-intervals-overlap}}
 	\label{appendix: overlap-proof}
 	\lemmaintervalnoskipping*
 	\begin{proof}{\textit{Proof.}}
 		The first and third inequalities are immediate from Lemma~\ref{lemma: ic-interval-decrease}. Hence our focus is to show that for any $K\ge 3$,
 		$$\underline{\mu}_{K-2} \le \overline{\mu}_K.$$ By the indifference conditions \eqref{eqn: under-mu} and \eqref{eqn: over-mu}, we can instead show $\frac{1-q}{q} \frac{\mathcal{B}_{K-2}}{\mathcal{G}_{K-2}} < \frac{q}{1-q} \frac{\mathcal{B}_K}{\mathcal{G}_K}$,
 		or equivalently
 		\begin{align*}
 			\left(\frac{1-q}{q}\right)^2 < \frac{\mathcal{B}_K}{\mathcal{G}_K}\frac{\mathcal{G}_{K-2}}{\mathcal{B}_{K-2}},
 		\end{align*}
 		where $\mathcal{G}_K$ (respectively, $\mathcal{B}_K$) are defined in Section~\ref{section: single-batch} as the probability that the object gets allocated to some agent in the batch of size $K$ conditional on the true quality of the object being good (respectively, bad). We show the last inequality in three steps. Each step involves finding a tighter lower bound for $\frac{\mathcal{B}_K}{\mathcal{G}_K}\frac{\mathcal{G}_{K-2}}{\mathcal{B}_{K-2}}$ than the previous step. 
 		
 		First of all, we bound $\frac{\mathcal{B}_K}{\mathcal{G}_K}\frac{\mathcal{G}_{K-2}}{\mathcal{B}_{K-2}}$ using a function of binomial distributions. 
 		
 		Recall the indifference conditions \eqref{eqn: under-mu} and \eqref{eqn: over-mu}:
 		\begin{align*}
 			\frac{\underline{\mu}_K}{1-\underline{\mu}_K} = \frac{1-q}{q} \frac{\mathcal{B}_K}{\mathcal{G}_K}, &\quad
 			\frac{\overline{\mu}_K}{1-\overline{\mu}_K} = \frac{q}{1-q} \frac{\mathcal{B}_K}{\mathcal{G}_K}.
 			\normalsize
 		\end{align*}
 		Rearranging the terms yields
 		\begin{align*}
 			\frac{\mathcal{B}_K}{\mathcal{G}_K} &= \frac{\underline{\mu}_K}{1-\underline{\mu}_K} \frac{q}{1-q} = \frac{1-\bbP\left(X_K \ge \frac{K+1}{2}\right)}{\bbP\left(X_K \ge \frac{K+1}{2}\right)} \frac{q}{1-q} \\
 			\frac{\mathcal{G}_{K-2}}{\mathcal{B}_{K-2}} &= \frac{1-\underline{\mu}_{K-2}}{\underline{\mu}_{K-2}} \frac{1-q}{q} = \frac{\bbP\left(X_{K-2} \ge \frac{K+1}{2} - 2\right)}{1-\bbP\left(X_{K-2} \ge \frac{K+1}{2} - 2\right)} \frac{1-q}{q}.
 		\end{align*}
 		Therefore,
 		\begin{align*}
 			\frac{\mathcal{B}_K}{\mathcal{G}_K}\frac{\mathcal{G}_{K-2}}{\mathcal{B}_{K-2}} &= \frac{1-\bbP\left(X \ge \frac{K+1}{2}\right)}{\bbP\left(X \ge \frac{K+1}{2}\right)}
 			\frac{\bbP\left(X_{K-2} \ge \frac{K+1}{2} - 2\right)}{1-\bbP\left(X_{K-2} \ge \frac{K+1}{2} - 2\right)} \\
 			&= \frac{\sum_{y=0}^{\frac{K+1}{2}-1} {K \choose y} q^y (1-q)^{K-y}}{\sum_{y=0}^{\frac{K+1}{2}-1} {K \choose y} (1-q)^y q^{K-y}} \frac{\sum_{y=0}^{\frac{K+1}{2}-2} {K-2 \choose y} (1-q)^y q^{K-2-y}}{\sum_{y=0}^{\frac{K+1}{2}-2} {K-2 \choose y} q^y (1-q)^{K-2-y}}.
 		\end{align*}
 		Observe that the last expression can be written as
 		\footnotesize
 		\begin{align*}
 			\left(\frac{\sum_{y=0}^{\frac{K+1}{2}-2} {K \choose y} q^y (1-q)^{K-y} 
 				+ {K \choose \frac{K+1}{2}-1} q^{\frac{K+1}{2}-1} (1-q)^{\frac{K+1}{2}}}{\sum_{y=0}^{\frac{K+1}{2}-2} {K \choose y} (1-q)^y q^{K-y}
 				+ {K \choose \frac{K+1}{2}-1} (1-q)^{\frac{K+1}{2}-1} q^{\frac{K+1}{2}}}\right)
 			\left(
 			\frac{\sum_{y=0}^{\frac{K+1}{2}-2} {K-2 \choose y} (1-q)^y q^{K-2-y}}    {\sum_{y=0}^{\frac{K+1}{2}-2} {K-2 \choose y} q^y (1-q)^{K-2-y}}\right)
 		\end{align*}
 		\normalsize
 		which is lower bounded by (by Lemma~\ref{property: supp-claim-shrink})
 		\begin{align*}
 			\frac{\sum_{y=0}^{\frac{K+1}{2}-2} {K \choose y} q^y (1-q)^{K-y}}{\sum_{y=0}^{\frac{K+1}{2}-2} {K \choose y} (1-q)^y q^{K-y}}
 			\frac{\sum_{y=0}^{\frac{K+1}{2}-2} {K-2 \choose y} (1-q)^y q^{K-2-y}}    {\sum_{y=0}^{\frac{K+1}{2}-2} {K-2 \choose y} q^y (1-q)^{K-2-y}}.
 		\end{align*}
 		Hence 
 		\begin{align}
 			\frac{\mathcal{B}_K}{\mathcal{G}_K} > \frac{\sum_{y=0}^{\frac{K+1}{2}-2} {K \choose y} q^y (1-q)^{K-y}}{\sum_{y=0}^{\frac{K+1}{2}-2} {K \choose y} (1-q)^y q^{K-y}}
 			\frac{\sum_{y=0}^{\frac{K+1}{2}-2} {K-2 \choose y} (1-q)^y q^{K-2-y}}    {\sum_{y=0}^{\frac{K+1}{2}-2} {K-2 \choose y} q^y (1-q)^{K-2-y}}.
 			\label{eqn: b/g-step-1}
 		\end{align}
 		
 		The second step involves rearranging the lower bound in \eqref{eqn: b/g-step-1} found in the previous step. We introduce the following notations. 
 		\begin{align*}
 			\eta_q &\triangleq \sum_{y=0}^{\frac{K+1}{2}-2} {K \choose y} q^y (1-q)^{K-y} {K-2 \choose y} (1-q)^y q^{K-2-y}\\
 			\eta_{1-q} &\triangleq \sum_{y=0}^{\frac{K+1}{2}-2} {K \choose y} (1-q)^y q^{K-y} {K-2 \choose y} q^y (1-q)^{K-2-y}\\
 			\psi_q &\triangleq \sum_{y=0}^{\frac{K+1}{2}-2}  \sum_{z=y+1}^{\frac{K+1}{2}-2} \Bigg[{K \choose y} q^y (1-q)^{K-y} {K-2 \choose z} (1-q)^z q^{K-2-z}\\
 			&\qquad \qquad \qquad+ {K \choose z} q^z (1-q)^{K-z} {K-2 \choose y} (1-q)^y q^{K-2-y} \Bigg]\\
 			\psi_{1-q} &\triangleq \sum_{y=0}^{\frac{K+1}{2}-2} \sum_{z=y+1}^{\frac{K+1}{2}-2} \Bigg[{K \choose y} (1-q)^y q^{K-y} {K-2 \choose z} q^z (1-q)^{K-2-z}\\
 			&\qquad \qquad \qquad + {K \choose z} (1-q)^z q^{K-z} {K-2 \choose y} q^y (1-q)^{K-2-y} \Bigg].
 		\end{align*}   
 		\normalsize
 		Then the right-hand side term of \eqref{eqn: b/g-step-1} is equal to $\frac{\eta_q + \psi_q}{\eta_{1-q} + \psi_{1-q}},$ and thus
 		\begin{align*}
 			\frac{\mathcal{B}_K}{\mathcal{G}_K}\frac{\mathcal{G}_{K-2}}{\mathcal{B}_{K-2}}  
 			&= \frac{\eta_q + \psi_q}{\eta_{1-q} + \psi_{1-q}}. 
 		\end{align*}    
 		Since $\eta_q + \psi_q$ and $\eta_{1-q} + \psi_{1-q}$ are series of the same length $\frac{K^2-1}{8} \left(=\frac{(\frac{K+1}{2}-1) \frac{K+1}{2}}{2} \right),$ we use Lemma~\ref{property: sum-min} twice to find the following lower bound. 
 		\begin{align*}
 			\frac{\eta_q + \psi_q}{\eta_{1-q} + \psi_{1-q}} 
 			\ge \min& \left\{\frac{\eta_q}{\eta_{1-q}}, \frac{\psi_q}{\psi_{1-q}} \right\}\\
 			\ge \min\Bigg\{
 			\min_{y \in \{0,\ldots, \frac{K+1}{2}-2\}}  &\frac{{K \choose y}{K-2 \choose y} q^{K-2} (1-q)^{K}}{{K \choose y}{K-2 \choose y} (1-q)^{K-2} q^{K}}, \\
 			\min_{\substack{(y,z) | y \in \{0,\ldots, \frac{K+1}{2}-2\},\\ \quad z \in \{y+1,\ldots, \frac{K+1}{2}-2\}}}&  \frac{{K \choose y} {K-2 \choose z} q^{K-2+y-z} (1-q)^{K-y+z}  + {K \choose z} {K-2 \choose y} q^{K-2-y+z} (1-q)^{K+y-z}}{{K \choose y} {K-2 \choose z} (1-q)^{K-2+y-z} q^{K-y+z} + {K \choose z} {K-2 \choose y} (1-q)^{K-2-y+z} q^{K+y-z}}
 			\Bigg\}
 		\end{align*}
 		where both inequalities follow by Lemma~\ref{property: sum-min}.

 		The final step involves finding the lower bounds of these terms. For any $y$ we have 
 		$$\frac{{K \choose y}{K-2 \choose y} q^{K-2} (1-q)^{K}}{{K \choose y}{K-2 \choose y} (1-q)^{K-2} q^{K}} = \left(\frac{1-q}{q}\right)^2$$ and thus $$\min_{y \in \{0,\ldots, \frac{K+1}{2}-2\}} \frac{{K \choose y}{K-2 \choose y} q^{K-2} (1-q)^{K}}{{K \choose y}{K-2 \choose y} (1-q)^{K-2} q^{K}} = \left(\frac{1-q}{q}\right)^2.$$
 		Similarly for any $y \in \{0,\ldots, \frac{K+1}{2}-2\}, z \in \{y+1,\ldots, \frac{K+1}{2}-2\},$ 
 		\begin{align*}
 			&\frac{{K \choose y} {K-2 \choose z} q^{K-2+y-z} (1-q)^{K-y+z}  + {K \choose z} {K-2 \choose y} q^{K-2-y+z} (1-q)^{K+y-z}}{{K \choose y} {K-2 \choose z} (1-q)^{K-2+y-z} q^{K-y+z} + {K \choose z} {K-2 \choose y} (1-q)^{K-2-y+z} q^{K+y-z}}\\ 
 			=& \frac{q^{K-2+y-z}(1-q)^{K+y-z}\left[{K \choose y}{K-2 \choose z}(1-q)^{2(z-y)} + {K \choose z}{K-2 \choose y}q^{2(z-y)} \right]}
 			{(1-q)^{K-2+y-z}q^{K+y-z}\left[{K \choose y}{K-2 \choose z}q^{2(z-y)} + {K \choose z}{K-2 \choose y}(1-q)^{2(z-y)} \right]} \\
 			=& \left(\frac{1-q}{q}\right)^2\frac{{K \choose y}{K-2 \choose z}(1-q)^{2(z-y)} + {K \choose z}{K-2 \choose y}q^{2(z-y)}}
 			{{K \choose y}{K-2 \choose z}q^{2(z-y)} + {K \choose z}{K-2 \choose y}(1-q)^{2(z-y)}} \\
 			>& \left(\frac{1-q}{q} \right)^2
 		\end{align*}
 		
 		To show the last inequality above, because $z > y$ and $q > 1-q,$ it is sufficient to show that 
 		\begin{align}
 			{K \choose z}{K-2 \choose y} > {K \choose y}{K-2 \choose z}.
 			\label{eqn: b/g-step-3}
 		\end{align}
 		Observe that 
 		\begin{align*}
 			{K \choose z}{K-2 \choose y} &= \frac{K!}{z! (K-z)!} \frac{(K-2)!}{y! (K-y-2)!}\\
 			{K \choose y}{K-2 \choose z} &= \frac{K!}{y! (K-y)!} \frac{(K-2)!}{z! (K-z-2)!}.
 		\end{align*}
 		Then \eqref{eqn: b/g-step-3} holds since
 		\begin{align*}
 			\frac{1}{(K-z)! (K-y-2)!} &> \frac{1}{(K-y)! (K-z-2)!}
 		\end{align*}
 		because $$(K-y-1)(K-y) = \frac{(K-y)!}{(K-y-2)!} > \frac{(K-z)!}{(K-z-2)!} = (K-z-1)(K-z).$$
 		
 		Therefore this concludes the proof of Lemma~\ref{lemma: ic-intervals-overlap}:
 		\begin{align*}
 			\frac{\mathcal{B}_K}{\mathcal{G}_K}\frac{\mathcal{G}_{K-2}}{\mathcal{B}_{K-2}} > \frac{\eta_q + \psi_q}{\eta_{1-q} + \psi_{1-q}} > \left(\frac{1-q}{q}\right)^2
 		\end{align*}
 		where the first lower bound was achieved via the first step, and the final tighter lower bound was achieved via the second and third steps.
 	\end{proof}

 	\section{Technical Properties}
 	\label{section: supp-proofs}
 	\begin{restatable}{lemma}{propertycombinatorial}
 		$\frac{1}{y}{K-1 \choose y-1} = \frac{1}{K}{K \choose y}$.
 		\label{property: combinatorial}
 	\end{restatable}
 	\begin{proof}{\textit{Proof.}}
 		\begin{align*}
 			\frac{1}{y}{K-1 \choose y-1} &= \frac{1}{y} \frac{(K-1)!}{(y-1)! (K-y)!} \\
 			&= \frac{(K-1)!}{y! (K-y)!} \\
 			&= \frac{1}{K} \frac{K!}{y! (K-y)!} \\
 			&= \frac{1}{K} {K \choose y}.
 		\end{align*}
 	\end{proof}
 	
 	\begin{restatable}{lemma}{propertybinomialtail} 
 		For any $q \in (0.5,1)$ and odd $K \in \mathbb{N},$ let $X_K$ be a Binomial$(K,q)$ random variable.
 		$$ \bbP\left(X_{K+2} \ge \frac{(K+2)+1}{2} \right) > \bbP\left(X_K \ge \frac{K+1}{2} \right) $$ 
 		and as $K$ grows
 		$$ \lim_{K \rightarrow \infty} \bbP\left(X_K \ge \frac{K+1}{2}\right) = 1.$$
 		\label{property: tail-increase}
 	\end{restatable}
 	\begin{proof}{\textit{Proof.}}
 		To show the first part, we use the following property of the binomial distribution $$\bbP \left( X_{K+1} \ge \frac{K+1}{2}+1 \right) =  \bbP \left( X_{K} \ge \frac{K+1}{2}+1 \right) + q \bbP \left( X_{K} \ge \frac{K+1}{2} \right).$$ Then in follows that
 		\begin{small}
 			\begin{align*}
 				\bbP\left(X_{K+2} \ge \frac{(K+2)+1}{2} \right) &=  \bbP\left(X_{K+1} \ge \frac{K+1}{2} + 1\right) + q \bbP\left(X_{K+1} \ge \frac{K+1}{2} \right)  \\
 				=&  \bbP\left(X_{K} \ge \frac{K+1}{2} + 1\right) + q \bbP\left(X_{K} \ge \frac{K+1}{2} \right)\\
 				&\quad + q \left[ \bbP\left(X_{K} \ge \frac{K+1}{2} \right) 
 				+ q \bbP\left(X_{K} \ge \frac{K+1}{2}-1 \right)  \right] \\
 				= 2q &\bbP\left(X_{K} \ge \frac{K+1}{2} \right) + \bbP\left(X_{K} \ge \frac{K+1}{2} + 1\right) + q^2 \bbP\left(X_{K} \ge \frac{K+1}{2}-1 \right) \\
 				>& \bbP\left(X_{K} \ge \frac{K+1}{2} \right)  + \bbP\left(X_{K} \ge \frac{K+1}{2} + 1\right)  + q^2 \bbP\left(X_{K} \ge \frac{K+1}{2}-1 \right)
 			\end{align*}
 		\end{small}
 		and we obtain the desired inequality
 		\begin{align*}
 			\bbP\left(X_{K+2} \ge \frac{(K+2)+1}{2} \right) &> \bbP\left(X_{K} \ge \frac{K+1}{2} \right).
 		\end{align*}
 		To show the second part of the property, note that by the Weak Law of Large Numbers, for any $\epsilon > 0$,
 		\begin{align*}
 			\lim_{K \rightarrow \infty} \bbP \left( |X_K/K - q| < \epsilon\right) = 1,
 		\end{align*}
 		and because $q - X_K/K \le |X_K/K - q|,$ 
 		\begin{align*}
 			1 = \lim_{K \rightarrow \infty} \bbP \left( |X_K/K - q| < \epsilon\right) \le \lim_{K \rightarrow \infty} \bbP \left( q-X_K/K < \epsilon\right).
 		\end{align*}
 		Since $\bbP \left( q-X_K/K < \epsilon\right) \le 1$ it follows that
 		\begin{align}
 			\lim_{K \rightarrow \infty} \bbP \left( q - \epsilon < X_K/K \right) = 1. 
 			\label{eqn: wlln}
 		\end{align}
 		Because $q > 1/2,$ there exists some $\epsilon > 0$ such that $q-\epsilon = \frac{1}{2}$ and that satisfies Equation \eqref{eqn: wlln}: 
 		\begin{align*}
 			1 = \lim_{K \rightarrow \infty} \bbP \left( \frac{1}{2} < \frac{X_K}{K} \right) &= \lim_{K \rightarrow \infty} \bbP \left( X_K > \frac{K}{2} \right) =  \lim_{K \rightarrow \infty} \bbP \left( X_K \ge \frac{K+1}{2} \right).
 		\end{align*}
 	\end{proof}
 	
 	\begin{lemma}
 		For any $n$ and non-negative sequences $\{a_i\}_{i=0}^n$ and $\{b_i\}_{i=0}^n,$ $$\min_{i} \frac{a_i}{b_i} \le \frac{\sum_{i=0}^n a_i}{\sum_{i=0}^n b_i}.$$
 		\label{property: sum-min}
 	\end{lemma}
 	\begin{proof}{\textit{Proof.}}
 		Let $r_{min} \triangleq \min_i \frac{a_i}{b_i}.$ Then for any $i,$ $\frac{a_i}{b_i} \ge r_{min}$, or equivalently 
 		$$a_i \ge r_{min}\, b_i.$$ Hence, $$\frac{\sum_{i=0}^n a_i}{\sum_{i=0}^n b_i} \ge \frac{\sum_{i=0}^n r_{min} b_i}{\sum_{i=0}^n b_i} = r_{min}.$$
 	\end{proof} 
 	
 	\begin{lemma}
 		For any $K \ge 3,$ 
 		$$\frac{\sum_{y=0}^{\frac{K+1}{2}-2} {K \choose y} q^y (1-q)^{K-y} 
 			+ {K \choose \frac{K+1}{2}-1} q^{\frac{K+1}{2}-1} (1-q)^{\frac{K+1}{2}}}{\sum_{y=0}^{\frac{K+1}{2}-2} {K \choose y} (1-q)^y q^{K-y}
 			+ {K \choose \frac{K+1}{2}-1} (1-q)^{\frac{K+1}{2}-1} q^{\frac{K+1}{2}}} > \frac{\sum_{y=0}^{\frac{K+1}{2}-2} {K \choose y} q^y (1-q)^{K-y}}{\sum_{y=0}^{\frac{K+1}{2}-2} {K \choose y} (1-q)^y q^{K-y}}.$$ 
 		\label{property: supp-claim-shrink}
 	\end{lemma}
 	\begin{proof}{\textit{Proof.}}
 		It is equivalent to show that
 		\footnotesize
 		\begin{align*}
 			{K \choose \frac{K+1}{2}-1} q^{\frac{K+1}{2}-1} (1-q)^{\frac{K+1}{2}} &\sum_{y=0}^{\frac{K+1}{2}-2} {K \choose y} (1-q)^y q^{K-y}\\
 			> {K \choose \frac{K+1}{2}-1} q^{\frac{K+1}{2}} (1-q)^{\frac{K+1}{2}-1} &\sum_{y=0}^{\frac{K+1}{2}-2} {K \choose y} q^y (1-q)^{K-y}
 		\end{align*}
 		\normalsize
 		which reduces to
 		\begin{align*}
 			(1-q) \sum_{y=0}^{\frac{K+1}{2}-2} {K \choose y} (1-q)^y q^{K-y} &> q \sum_{y=0}^{\frac{K+1}{2}-2} {K \choose y} q^y (1-q)^{K-y}
 		\end{align*}
 		and equivalently
 		\begin{align*}
 			\frac{\sum_{y=0}^{\frac{K+1}{2}-2} {K \choose y} (1-q)^y q^{K-y}}{\sum_{y=0}^{\frac{K+1}{2}-2} {K \choose y} q^y (1-q)^{K-y}} &> \frac{q}{1-q}.
 		\end{align*}
 		The last inequality holds because by Lemma~\ref{property: sum-min},
 		\begin{align*}
 			\frac{\sum_{y=0}^{\frac{K+1}{2}-2} {K \choose y} (1-q)^y q^{K-y}}{\sum_{y=0}^{\frac{K+1}{2}-2} {K \choose y} q^y (1-q)^{K-y}} &\ge \min_y \left( \frac{q}{1-q} \right)^{K-2y} \\
 			&= \left( \frac{q}{1-q} \right)^{K-2\left(\frac{K+1}{2}-2 \right)} \\
 			&= \left( \frac{q}{1-q} \right)^3 \\
 			&> \frac{q}{1-q}.
 		\end{align*}
 	\end{proof}

\bibliographystyle{ACM-Reference-Format}
\bibliography{arxiv/arxiv2024}
\end{document}